\documentclass[preprint,11pt]{elsarticle}
\usepackage[nointegrals]{wasysym}   
\usepackage{standalone}

\usepackage{amsmath, amssymb, amsthm}
\theoremstyle{definition}

\newtheorem{theorem}{Theorem}

\usepackage{graphicx}
\usepackage{float}
\usepackage{subcaption}
\graphicspath{{figures/}}
\usepackage{booktabs}
\usepackage[table]{xcolor}
\usepackage{tikz}
\usetikzlibrary{calc, matrix, backgrounds, fit, arrows.meta}

\definecolor{edgegreen}{HTML}{6CBF70}        
\definecolor{xred}{HTML}{C0392B}             
\definecolor{chordorange}{HTML}{E8913A}      
\definecolor{edgeblue}{HTML}{6C8EBF}         
\definecolor{edgered}{HTML}{B85450}          
\definecolor{coreblue}{HTML}{DAE8FC}         
\definecolor{corebluestroke}{HTML}{6C8EBF}   
\definecolor{coreyellow}{HTML}{FFF2CC}       
\definecolor{coreyellowstroke}{HTML}{D6B656} 
\definecolor{hlblue}{HTML}{DAE8FC}           
\definecolor{hlred}{HTML}{F8CECC}            
\definecolor{pathband}{HTML}{D5E8D4}         
\definecolor{chordcell}{HTML}{FFE6CC}        

\tikzset{
  physnode/.style={circle, draw=corebluestroke,   line width=0.9pt, fill=coreblue,
                   inner sep=0pt, minimum size=6pt},
  intnode/.style={ circle, draw=coreyellowstroke, line width=0.9pt, fill=coreyellow,
                   inner sep=0pt, minimum size=6pt},
  tnbond/.style={line width=0.9pt, black!85},
  tnleg/.style={ line width=0.9pt, black!85},
  chordbond/.style={line width=0.9pt, chordorange},
  corelbl/.style={font=\scriptsize, inner sep=2pt},
  tracepath/.style={edgegreen, dashed, dash pattern=on 5pt off 3pt, line width=1.1pt,
                    -{Stealth[length=2.6mm,width=2.2mm]}, shorten >=2pt, shorten <=2pt},
  deadpath/.style={ xred, dashed, dash pattern=on 5pt off 3pt, line width=1.1pt,
                    -{Stealth[length=2.6mm,width=2.2mm]}, shorten >=2pt, shorten <=2pt},
}

\usepackage{geometry}
\usepackage{algorithm}
\usepackage{algpseudocode}

\algrenewcommand\algorithmiccomment[1]{%
    \hfill\(\triangleright\) \textit{\small #1}}

\usepackage{todonotes}

\begin{document}
\begin{frontmatter}

\title{Computing with traceable tensor networks}

\author[ucsc]{Sarah Ellwein}
\author[ucsc]{Daniele Venturi\corref{cor}}
\ead{venturi@ucsc.edu}

\address[ucsc]{Department of Applied Mathematics,
University of California Santa Cruz\\ Santa Cruz (CA) 95064}

\cortext[cor]{Corresponding author}

\journal{ArXiv}

\begin{abstract}
We introduce a new SVD-based tensor decomposition method for tensor networks with arbitrary graph topologies, extending classical hierarchical SVD-based techniques to networks with cycles and general connectivity. We also introduce addition and rounding procedures for traceable tensor graphs, enabling step-rounding time integration of high-dimensional PDEs directly in graph format, with rank truncation controlled to a prescribed tolerance at every time step. We demonstrate the new method on the decomposition of multivariate functions and on the numerical solution of the Fokker--Planck equation, and find that the graph-format representation attains comparable or better accuracy than the classical tensor train and hierarchical Tucker tensor formats, while using substantially fewer degrees of freedom at lower computational cost.
\end{abstract}


\end{frontmatter}

\section{Introduction}

Approximating high-dimensional problems via numerical algorithms usually runs into a fundamental challenge: the number of degrees of freedom (DOF) grows exponentially with the dimension of the problem. Classical examples include the approximation of high-dimensional functions and numerical solvers for high-dimensional partial differential equations (PDEs) , such as the Fokker–Planck equation \cite{risken1996,dektor2021, rodgers2023,tang2024}, the BBGKY hierarchy in kinetic theory \cite{boelens2020a, cercignani, dimarco2014, einkemmer2025}, and functional differential equations \cite{rodgers2024, venturi2021, venturi2018}. 

When the solution admits a low-rank structure, numerical tensor methods can mitigate this growth via low-rank approximations. A successful class of such methods are based on singular value decompositions (SVD), which construct each factor of the representation by applying a hierarchical sequence of SVDs to the target tensor. 
Two widely used tensor formats in this class, namely tensor train (TT) \cite{oseledets2011} and hierarchical Tucker (HT) \cite{grasedyck2010}, yield provable quasi-optimal low-rank approximations and reduce DOF growth from exponential to polynomial in dimension. 
SVD-based decompositions have been extended to other topologies, including the tensor ring (TR) \cite{zhao2016}, tensor wheel (TW) \cite{wangTW2024}, and the fully-connected tensor network (FCTN) \cite{zheng2021, wangFCTN2024}. FCTN in particular generalizes the chain and tree topologies of TT and HT to an arbitrary graph and admits a recursive SVD construction (FCTN-SVD) \cite{wangFCTN2024}.

These tensor formats extend naturally to high-dimensional PDEs, where the curse of dimensionality represents one of the main computational bottlenecks. As an example, consider an initial-value problem of the form
\begin{equation}
    \frac{\partial u}{\partial t} = \mathcal{L}(u),
    \qquad u(\mathbf{x}, 0) = u_0,
    \label{eq:pde}
\end{equation}
where $u(\mathbf{x}, t)$ depends on $d$ variables $\mathbf{x}$ and $\mathcal{L}$ is a linear differential operator, possibly depending on $\mathbf{x}$. Discretizing in space on a tensor-product grid yields a semi-discrete system $\dot{\mathbf{u}} = \mathbf{F}(\mathbf{u})$, where $\mathbf{u}(t) \in \mathbb{R}^{n_1 \times \cdots \times n_d}$ is a $d$-mode tensor whose full storage cost is prohibitive even for moderate $d$ and $n_j$.  For instance, with $n_j = 200$ points per dimension and $d = 6$, one has $64 \times 10^{12}$ (64 trillion) grid points, requiring about 512 terabytes of storage in double precision.

To overcome this problem, one may consider a low-rank representation of $\mathbf{u}(t)$ in a chosen tensor format throughout the simulation. One such method is {\em step-truncation} (or step-rounding), in which the tensor solution undergoes a time step forward using any conventional time-stepping scheme, followed by a recompression to remain on a low-rank manifold \cite{rodgers2022,dektor2021,venturi2018}.
As an example, the forward Euler step-truncation scheme takes the form
\begin{equation}
    \mathbf{u}^{n+1}
    \;=\; \mathfrak{R}_{\varepsilon}\!\left(
        \mathbf{u}^{n} + \Delta t \,  \mathbf{F}(\mathbf{u}^{n})
    \right),
    \label{eq:step-truncation}
\end{equation}
where $\mathfrak{R}_{\varepsilon}(\cdot)$ is a rounding operator that projects the input back onto the low-rank manifold up to a prescribed tolerance $\varepsilon$. Two primitive operations are required for this scheme to operate: addition, since $\mathbf{u}^{n} + \Delta t\,  \mathbf{F}(\mathbf{u}^{n})$ has rank up to the sum of the ranks of its summands, and rounding $\mathfrak{R}_{\varepsilon}$, which restores low-rankness after each step. 
Both TT and HT tensor formats satisfy these requirements: the addition is performed by a block-diagonal concatenation of factors, and $\mathfrak{R}_{\varepsilon}$ is implemented as an SVD-based sweep along the structure of the format \cite{oseledets2011, grasedyck2010}. A similar SVD-based rounding has been developed for the Tensor Ring (TR) format, where addition uses max-concatenation in place of block-diagonal concatenation, allowing an SVD sweep to compress inflated ranks \cite{mickelin2020}. 
While various algorithms exist for the SVD-based construction of low-rank tensor network representations of a $d$-dimensional full tensor, to our knowledge, no SVD-sweep based rounding procedures are available for tensor network formats other than TT, HT, or TR.

In this paper, we develop new SVD-based algorithm to compute arbitrary graph tensor network (GTN) decompositions with user-specified network topology---given the topology admits a graph structure---we call GTN-SVD. For the subclass of {\em traceable graphs}--those whose network admits a traceable path--we further introduce addition and rounding operations inspired by their TT counterparts \cite{oseledets2011}. Together, these operations close the step-truncation workflow on traceable GTNs, allowing time-evolving low-rank solutions to be carried in graph topologies beyond classic formats like TT, HT or TR. The rounding procedure exploits the traceable path ordering to reduce the SVD sweep to a TT-style pass. We demonstrate the resulting integrator on a four-dimensional Fokker-Planck equation, where the traceable GTN representation achieves comparable accuracy to TT and HT at lower DOF and reduced per-step rounding time.

As a preview of what a GTN decomposition can offer, consider the 4D multivariate function 
\begin{equation}
\begin{aligned}
f(\mathbf{x}) =
&\exp\!\Big(\!-4\Big[\cos(x_1-\pi)+\cos(x_3-\pi)-0.5\big)^2
   + 1.5\,\sin(x_1-\pi)\sin(x_3-\pi)+\\[2pt]
&1.6\,\cos(x_2-x_4) + 1.2\,\cos\!\big(2(x_2+x_4)\Big]\Big),
\end{aligned}
\label{eq:spotlight}
\end{equation}

\noindent
In Figure~\ref{fig:spotlight_comp5} we compare the performance of our GTN-SVD algorithm in a barbell topology (GTN-BB) against TT on this function. The two formats attain the same accuracy at very different cost: at a prescribed tolerance $\varepsilon = 10^{-8}$ the barbell requires $8,802$ degrees of freedom against $3.37\times 10^{6}$ for TT, a factor of $382$ fewer, and this factor grows to $498$ at $\varepsilon = 10^{-10}$. The same advantage carries over to time integration. Applying the two formats to a heat diffusion equation with a step truncation scheme, GTN-BB holds one to two orders of magnitude fewer degrees of freedom than TT throughout the evolution while meeting the same per-step tolerance. 
These preliminary results, detailed later in this paper, motivate the central question addressed here: how to decompose, add, and round tensors in a traceable graph format so as to enable efficient computation on low-rank tensor manifolds.

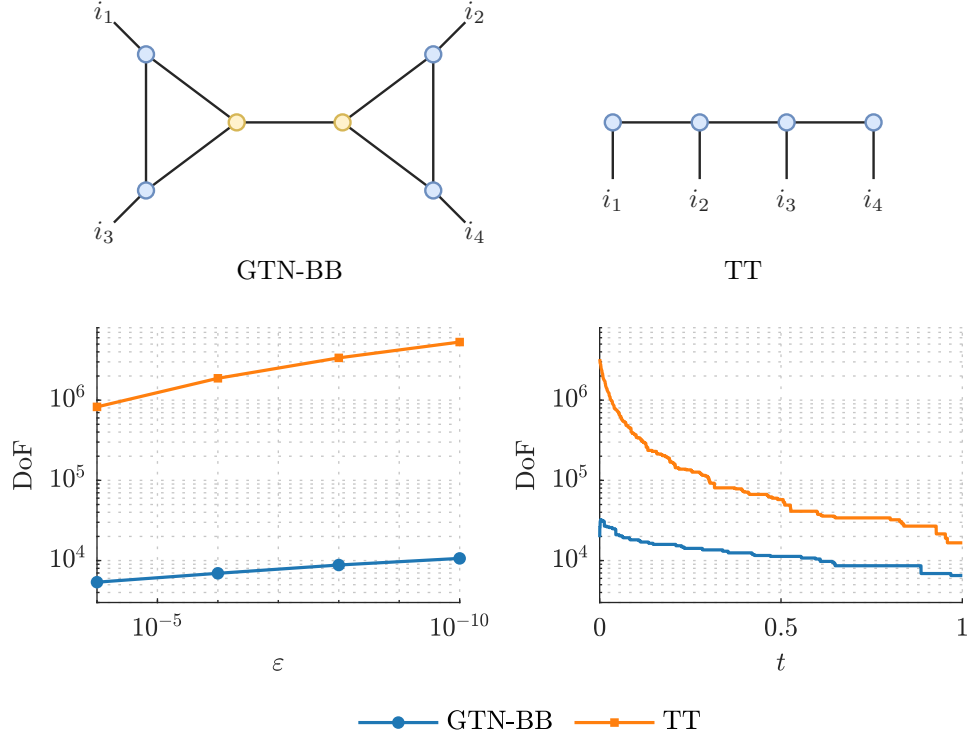
\begin{figure}[t]
\centering
\begin{tikzpicture}[scale=1.0, >=latex]
  \begin{scope}[shift={(0,0)}]
    \coordinate (H1) at (-1.9, 0.9); \coordinate (H3) at (-1.9,-0.9);
    \coordinate (S1) at (-0.7, 0);   \coordinate (S2) at ( 0.7, 0);
    \coordinate (H2) at ( 1.9, 0.9); \coordinate (H4) at ( 1.9,-0.9);
    \draw[tnbond] (H1)--(H3); \draw[tnbond] (H1)--(S1); \draw[tnbond] (H3)--(S1);
    \draw[tnbond] (S1)--(S2);
    \draw[tnbond] (H2)--(S2); \draw[tnbond] (H4)--(S2); \draw[tnbond] (H2)--(H4);
    \draw[tnleg] (H1) -- ++(135:0.6) node[font=\small, shift={(135:0.18)}]{$i_1$};
    \draw[tnleg] (H3) -- ++(225:0.6) node[font=\small, shift={(225:0.18)}]{$i_3$};
    \draw[tnleg] (H2) -- ++(45:0.6)  node[font=\small, shift={(45:0.18)}]{$i_2$};
    \draw[tnleg] (H4) -- ++(315:0.6) node[font=\small, shift={(315:0.18)}]{$i_4$};
    \node[physnode] at (H1) {}; \node[physnode] at (H3) {};
    \node[physnode] at (H2) {}; \node[physnode] at (H4) {};
    \node[intnode]  at (S1) {}; \node[intnode]  at (S2) {};
    \node[font=\small] at (0,-1.95) {GTN-BB};
  \end{scope}
  \begin{scope}[shift={(6.0,0)}]
    \foreach \i in {1,...,4}{\coordinate (T\i) at ({1.15*\i-2.875},0);}
    \draw[tnbond] (T1)--(T2)--(T3)--(T4);
    \foreach \i in {1,...,4}{
      \draw[tnleg] (T\i) -- ++(0,-0.75) node[below,font=\small]{$i_{\i}$};}
    \foreach \i in {1,...,4}{\node[physnode] at (T\i) {};}
    \node[font=\small] at (0,-1.95) {TT};
  \end{scope}
\end{tikzpicture}\\[2.5ex]
\includegraphics{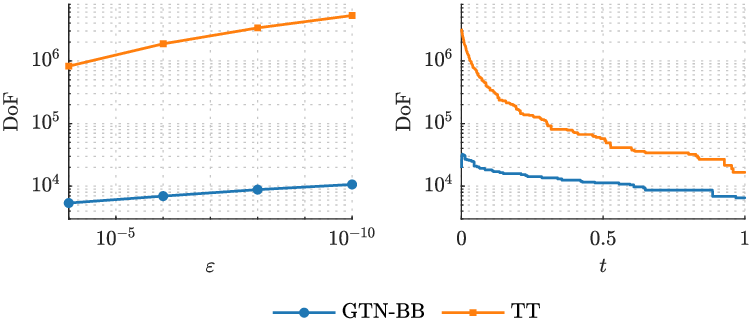}
\caption{Performance of our GTN-SVD algorithm on a barbell topology (GTN-BB) versus TT. Top: the two graph topologies, GTN-BB (left) and TT (right). Bottom: degrees of freedom for the tensor representation of the function \eqref{eq:spotlight} versus prescribed tolerance $\varepsilon$ (left), and versus time $t$ for the time-dependent solution of a four-dimensional diffusion PDE on a torus (right).}
\label{fig:spotlight_comp5}
\end{figure}

The remainder of the paper is organized as follows. In Section \ref{sec:SVD-GTN}, we establish the graph representation of tensor networks and introduce a new SVD-based algorithm that decomposes a tensor into an arbitrary user-defined graph tensor format. In Section \ref{sec:operations}, we define the essential operations---addition, Hadamard product, and rounding---that act on an arbitrary traceable GTN decomposition while remaining in its format, without reconstructing the full tensor. In Section \ref{sec:numerical_experiments} we present numerical examples applying the proposed algorithms to multivariate functions and to the 4D Fokker--Planck equation. The main findings are summarized in Section 
\ref{sec:summary}.

\section{SVD-Based Algorithms for Graph Tensor Networks}
\label{sec:SVD-GTN}

In this section, we introduce the tensor notation used throughout the paper, review the representation of tensor networks as graphs, and define the rank adjacency matrix, a compact encoding of a tensor network's topology and dimensions. Building on this, we then present our GTN-SVD algorithm, an SVD-based decomposition algorithm that produces a tensor decomposition in a user-defined graph tensor format.

\subsection{Tensor Notation and Basic Tensor Operations}
Throughout, we denote tensors by calligraphic letters ($\mathcal{X}$) and matrices by bold capitals ($\mathbf{X}$). We describe here the two important tensor operations used throughout the paper: matricization, which reshapes a tensor into a matrix, and the $k$-mode product, which applies a matrix to a single mode of a tensor.

\paragraph{Matricization} Reshaping a tensor into a matrix amounts to merging 
several of its indices into
one. Given $i_k \in \{1,\ldots,n_k\}$ for $k = 1,\ldots,d$, we write the resulting
\emph{multi-index} as
\begin{equation}
    \overline{i_1i_2\ldots i_d} = 1 + \sum_{k=1}^d(i_k-1)\prod_{j=1}^{k-1}n_j ,
    \label{eq:multi-index}
\end{equation}
\noindent
which enumerates the merged indices in first-index-fastest order. Two reshapings of
$\mathcal{X}$ recur. The \emph{single $k$-mode unfolding}
$\mathbf{X}_{(k)} \in \mathbb{R}^{\,n_k \times \prod_{i \neq k} n_i}$ isolates one
mode against all the others,
\begin{equation}
    \mathcal{X}(i_1,\ldots,i_k,\ldots,i_d)
    = \mathbf{X}_{(k)}\big(i_k,\ \overline{i_1\ldots i_{k-1}i_{k+1}\ldots i_d}\big) ,
    \label{eq:single-unfold}
\end{equation}
\noindent
placing mode $k$ along the rows and the merged remaining modes along the columns.
The \emph{$k$-mode unfolding} $\mathbf{X}_{\langle k\rangle} \in
\mathbb{R}^{\,\prod_{i \le k} n_i \,\times\, \prod_{i > k} n_i}$ instead splits the
modes into a leading and a trailing multi-index,
\begin{equation}
    \mathcal{X}(i_1,\ldots,i_k,\ldots,i_d)
    = \mathbf{X}_{< k>}\big(\overline{i_1\ldots i_{k}},\ \overline{i_{k+1}\ldots i_d}\big) ,
    \label{eq:unfold}
\end{equation}
\noindent
separating the first $k$ modes of $\mathcal{X}$ from the remaining $d-k$.

\paragraph{$k$-mode product}  The \emph{$k$-mode product} applies a matrix along a single mode, leaving the others
untouched. For $\mathbf{M}\in\mathbb{R}^{m\times n_k}$, the $k$-mode product
$\mathcal{X}\times_{k}\mathbf{M} \in
\mathbb{R}^{\,n_1\times\cdots\times n_{k-1}\times m\times n_{k+1}\times\cdots\times n_d}$
is given by
\begin{equation}
    (\mathcal{X}\times_{k}\mathbf{M})(i_1,\ldots,i_{k-1},j,i_{k+1},\ldots,i_d)=
    \sum_{i_k=1}^{n_k}\mathcal{X}(i_1,\ldots,i_k,\ldots,i_d)\,\mathbf{M}(j,i_k) ,
    \label{eq:k-mode-product}
\end{equation}
\noindent
so that mode $k$ is contracted against the columns of $\mathbf{M}$ and replaced by its rows. Equivalently, $\mathcal{Y} = \mathcal{X}\times_{k}\mathbf{M}$ is the tensor whose unfolding satisfies $\mathbf{Y}_{(k)} = \mathbf{M}\,\mathbf{X}_{(k)}$.

\paragraph{Tensor contraction} Given $\mathcal{X} \in \mathbb{R}^{n_1 \times \cdots \times n_p}$ and $\mathcal{Y} \in \mathbb{R}^{m_1 \times \cdots \times m_q}$ sharing a common dimension along mode $k$ of $\mathcal{X}$ and mode $l$ of $\mathcal{X}$ (i.e., $n_k = m_l$), their contraction along this pair of modes is the tensor $\mathcal{Z} \in \mathbb{R}^{n_1 \times \cdots \times n_{k-1} \times n_{k+1} \times \cdots \times n_p \times m_1 \times \cdots \times m_{l-1} \times m_{l+1} \times \cdots \times m_q}$ defined as
\begin{equation}
    \mathcal{Z}(i_1,\ldots,i_{k-1},i_{k+1},\ldots,i_p,\,j_1,\ldots,j_{l-1},j_{l+1},\ldots,j_q)
    = \sum_{i_k=1}^{n_k} \mathcal{X}(i_1,\ldots,i_k,\ldots,i_p)\,\mathcal{Y}(j_1,\ldots,i_k,\ldots,j_q) ,
    \label{eq:contraction}
\end{equation}
\noindent
so that mode $k$ of $\mathcal{X}$ and mode $l$ of $\mathcal{Y}$ are summed away and the remaining modes of both tensors are concatenated into a single, higher-order tensor. The order of $\mathcal{Z}$ is thus $p + q - 2$, reduced by two for each pair of contracted modes. Equivalently, a contraction can be written as a matrix product between the matricizations of $\mathcal{X}$ and $\mathcal{Y}$ along the contracted modes, $\mathbf{Z} = \mathbf{X}_{(k)}^\top \mathbf{Y}_{(l)}$, reshaped back into a tensor.

\subsection{Tensor Network Representation as Graphs}

Let $\mathcal{X} \in \mathbb{R}^{n_1 \times \cdots \times n_d}$ be a $d$-dimensional tensor, and suppose $\mathcal{X}$ admits a low-rank approximation $\mathcal{X} \approx \tilde{\mathcal{X}}$, where $\tilde{\mathcal{X}}$ is built from the contraction of smaller factor tensors $\mathcal{G}_1, \ldots, \mathcal{G}_M$, $M \geq d$, called tensor \emph{cores}. 
The tensor $\tilde{\mathcal{X}}$ can be conveniently represented in terms of a {\em graph} whose nodes represent the cores $\mathcal{G}_j$ and whose edges represent indices that are summed over (contracted).
%
For example, the classical TT format \cite{oseledets2011} expresses elements of $\mathcal{X}$ as

\begin{equation}
    \tilde{\mathcal{X}}(i_1, \ldots, i_d)
    \;=\;
    \sum_{\alpha_1, \ldots, \alpha_{d-1}}
    \mathcal{G}_1(i_1, \alpha_1)\,
    \mathcal{G}_2(\alpha_1, i_2, \alpha_2)\,
    \cdots\,
    \mathcal{G}_d(\alpha_{d-1}, i_d),
    \label{eq:tt-form}
\end{equation}

\noindent
which corresponds to a path graph $P_d$. Closing the TT chain into a loop by adding a single edge between the first and last cores produces the TR format \cite{zhao2016}

\begin{equation}
    \tilde{\mathcal{X}}(i_1, \ldots, i_d)
    \;=\;
    \sum_{\alpha_0, \alpha_1, \ldots, \alpha_{d-1}}
    \mathcal{G}_1(\alpha_0, i_1, \alpha_1)\,
    \mathcal{G}_2(\alpha_1, i_2, \alpha_2)\,
    \cdots\,
    \mathcal{G}_d(\alpha_{d-1}, i_d, \alpha_0),
    \label{eq:tr-form}
\end{equation}

\noindent
which corresponds to a cycle graph $C_d$. The HT format \cite{grasedyck2010} index form representation can be depicted as a binary tree, the $d$ physical indices are attached as leaves, and each interior node carries a transfer tensor $\mathcal{B}_t$ that contracts the two child subtrees into a single edge:

\begin{equation}
    \mathcal{U}_t(i_{\mu(t)}, \alpha_t)
    \;=\;
    \sum_{\alpha_{t_\ell},\, \alpha_{t_r}}
    \mathcal{B}_t(\alpha_{t_\ell}, \alpha_{t_r}, \alpha_t)\,
    \mathcal{U}_{t_\ell}(i_{\mu(t_\ell)}, \alpha_{t_\ell})\,
    \mathcal{U}_{t_r}(i_{\mu(t_r)}, \alpha_{t_r}),
    \label{eq:ht-form}
\end{equation}

\noindent
with $\tilde{\mathcal{X}} = \mathcal{U}_{t_{\mathrm{root}}}$. The FCTN format \cite{zheng2021} places an edge between every pair of cores

\begin{equation}
    \tilde{\mathcal{X}}(i_1, \ldots, i_d)
    \;=\;
    \sum_{\{\alpha_{j,k}\}_{1 \le j < k \le d}}
    \prod_{k=1}^{d}
    \mathcal{G}_k\bigl(i_k,\; \{\alpha_{j,k}\}_{j \ne k}\bigr),
    \label{eq:fctn-form}
\end{equation}

\noindent
yielding the complete graph $K_d$. These tensor formats, i.e., TT, HT, TR, and FCTN, are depicted in Figure \ref{fig:gtn_formats_4d} for $d=4$. They are all particular instances of a broader category of \emph{graph tensor networks}, in which the cores and edge structure are arranged according to a graph. 
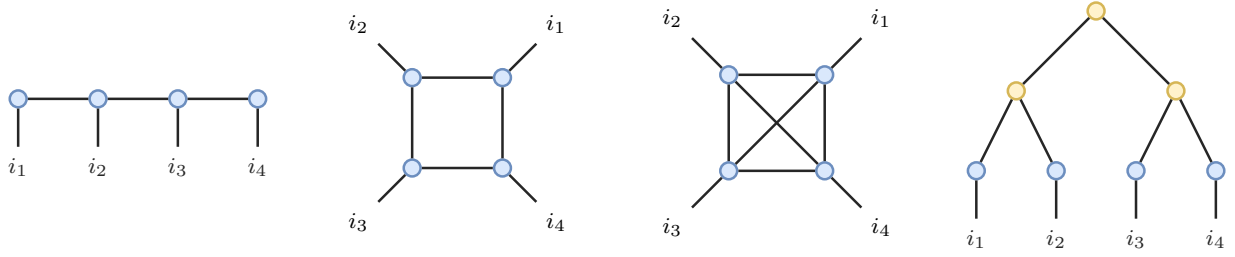
\begin{figure}[t]
    \centering
    \resizebox{\textwidth}{!}{%
    \begin{tikzpicture}[scale=1,>=latex]
    \begin{scope}[shift={(0,0)}]
        \foreach \i in {1,...,4} {
            \node[physnode] (TT\i) at ({\i-2.5}, 0.3) {};
        }
        \draw[tnbond] (TT1)--(TT2)--(TT3)--(TT4);         
        \foreach \i in {1,...,4} {                        
            \draw[tnleg] (TT\i) -- ++(0,-0.6) node[below,font=\scriptsize] {$i_{\i}$};
        }
    \end{scope}
    \begin{scope}[shift={(4,0)}]
        \def\rc{0.8}\def\rl{1.4}\def\rt{1.75}
        \foreach \i in {1,...,4} {
            \node[physnode] (TR\i) at ({45+90*(\i-1)}:\rc) {};
        }
        \draw[tnbond] (TR1)--(TR2)--(TR3)--(TR4)--(TR1);  
        \foreach \i in {1,...,4} {                        
            \draw[tnleg] (TR\i) -- ({45+90*(\i-1)}:\rl);
            \node[font=\scriptsize] at ({45+90*(\i-1)}:\rt) {$i_{\i}$};
        }
    \end{scope}
    \begin{scope}[shift={(8,0)}]
        \def\rc{0.85}\def\rl{1.5}\def\rt{1.85}
        \foreach \i in {1,...,4} {
            \node[physnode] (F\i) at ({45+90*(\i-1)}:\rc) {};
        }
        \draw[tnbond] (F1)--(F2)--(F3)--(F4)--(F1);       
        \draw[tnbond] (F1)--(F3); \draw[tnbond] (F2)--(F4); 
        \foreach \i in {1,...,4} {                        
            \draw[tnleg] (F\i) -- ({45+90*(\i-1)}:\rl);
            \node[font=\scriptsize] at ({45+90*(\i-1)}:\rt) {$i_{\i}$};
        }
    \end{scope}
    \begin{scope}[shift={(12,0)}]
        \node[physnode] (H1) at (-1.5,-0.6) {};           
        \node[physnode] (H2) at (-0.5,-0.6) {};
        \node[physnode] (H3) at ( 0.5,-0.6) {};
        \node[physnode] (H4) at ( 1.5,-0.6) {};
        \node[intnode]  (H5) at (-1, 0.4) {};             
        \node[intnode]  (H6) at ( 1, 0.4) {};
        \node[intnode]  (H7) at ( 0, 1.4) {};             
        \draw[tnbond] (H7)--(H5); \draw[tnbond] (H7)--(H6); 
        \draw[tnbond] (H5)--(H1); \draw[tnbond] (H5)--(H2);
        \draw[tnbond] (H6)--(H3); \draw[tnbond] (H6)--(H4);
        \foreach \i in {1,...,4} {                        
            \draw[tnleg] (H\i) -- ++(0,-0.6) node[below,font=\scriptsize] {$i_{\i}$};
        }
    \end{scope}
\end{tikzpicture}}
    \caption{Tensor train (TT), tensor ring (TR), fully-connected tensor network (FCTN), and hierarchical Tucker (HT) diagrams for $d=4$ (left to right). TT forms a path graph $P_4$, TR forms a cycle graph $C_4$, HT forms a binary tree, and FCTN forms a complete graph $K_4$.}
    \label{fig:gtn_formats_4d}
\end{figure}

To specify such a graph, we encode its topology in a \emph{rank adjacency matrix} $\mathbf{R}\in\mathbb{Z}^{p\times p}$, with $p\ge d$, whose entries record the physical mode sizes and the edge dimensions between every pair of cores, akin to an adjacency matrix used in graph theory. Specifically,  the rank adjacency matrix has entries 
\begin{equation}
    \mathbf{R}_{ij} =
    \begin{cases}
        n_i, & i = j, \\[2pt]
        r_{i,j} \quad \text{rank}(i,j) , & i \ne j, \text{ edge connects } i \text{ and } j. \\[2pt]
        1, & i \ne j, \text{ no edge connects } i \text{ and } j.
    \end{cases}
\end{equation}
The matrix is symmetric, as $ \mathbf{R}_{ij} =  \mathbf{R}_{ji}$. We write its diagonal entries as $ \mathbf{R}_{jj} = n_j$ for $j = 1, \ldots, p$, where $n_1, \ldots, n_d$ are the physical mode sizes of $\mathcal{X}$ and $n_j = 1$ for $d < j \leq p$, corresponding to internal nodes that are not associated with a physical index of the original tensor. In Figure~\ref{fig:tn_rm}, we provide examples of the rank adjacency matrix for three well-known tensor formats, namely FCTN, TT, and HT.

\begin{figure}[t]
    \centering
    \setlength{\tabcolsep}{10pt}\renewcommand{\arraystretch}{1.3}
\begin{tabular*}{0.75\textwidth}{@{\extracolsep{\fill}} l c c @{}}
\toprule
 & \textit{Tensor network} & \textit{Rank adjacency matrix} \\
\midrule
\textbf{FCTN} &
\begin{tikzpicture}[scale=0.6,>=latex, baseline={(current bounding box.center)}]
  \def\r{1.7}\def\leg{0.75}
  \foreach \i [evaluate=\i as \ang using 72*(\i-1)+90] in {1,...,5}{
    \node[physnode] (G\i) at (\ang:\r) {};
    \draw[tnleg] (G\i) -- (\ang:{\r+\leg});
    \node[font=\scriptsize] at (\ang:{\r+\leg+0.5}) {$i_{\i}$};
    \node[corelbl] at ({\ang+15}:{\r+0.42}) {$\mathcal{G}_{\i}$};   
  }
  \foreach \i in {1,...,5}{\foreach \j in {1,...,5}{\ifnum\i<\j \draw[tnbond] (G\i)--(G\j);\fi}}
\end{tikzpicture}
&
\begin{tikzpicture}[baseline={(R.center)}]
  \matrix (R) [matrix of math nodes, ampersand replacement=\&,
      left delimiter=(, right delimiter=),
      nodes={minimum size=0.6cm, inner sep=1pt, anchor=center, font=\footnotesize}]
  {
    n_1 \& r_{1,2} \& r_{1,3} \& r_{1,4} \& r_{1,5} \\
    r_{1,2} \& n_2 \& r_{2,3} \& r_{2,4} \& r_{2,5} \\
    r_{1,3} \& r_{2,3} \& n_3 \& r_{3,4} \& r_{3,5} \\
    r_{1,4} \& r_{2,4} \& r_{3,4} \& n_4 \& r_{4,5} \\
    r_{1,5} \& r_{2,5} \& r_{3,5} \& r_{4,5} \& n_5 \\
  };
  \foreach \k in {1,...,5}{
    \node[font=\scriptsize, left=10pt]  at (R-\k-1.west)  {$\mathcal{G}_{\k}$};
    \node[font=\scriptsize, above=1pt] at (R-1-\k.north) {$\mathcal{G}_{\k}$};}
\end{tikzpicture}
\\[6pt]
\textbf{TT} &
\begin{tikzpicture}[scale=0.6,>=latex, baseline={(current bounding box.center)}]
  \foreach \i in {1,...,4}{\node[physnode, label={[corelbl]above:$\mathcal{G}_{\i}$}] (G\i) at (1.4*\i,0) {};}
  \draw[tnbond] (G1)--(G2)--(G3)--(G4);
  \foreach \i in {1,...,4}{\draw[tnleg] (G\i) -- ++(0,-0.8) node[below,font=\scriptsize]{$i_{\i}$};}
\end{tikzpicture}
&
\begin{tikzpicture}[baseline={(R.center)}]
  \matrix (R) [matrix of math nodes, ampersand replacement=\&,
      left delimiter=(, right delimiter=),
      nodes={minimum size=0.5cm, inner sep=1pt, anchor=center, font=\footnotesize}]
  {
    n_1 \& r_{1,2} \& 1 \& 1 \\
    r_{1,2} \& n_2 \& r_{2,3} \& 1 \\
    1 \& r_{2,3} \& n_3 \& r_{3,4} \\
    1 \& 1 \& r_{3,4} \& n_4 \\
  };
  \foreach \k in {1,...,4}{
    \node[font=\scriptsize, left=10pt]  at (R-\k-1.west)  {$\mathcal{G}_{\k}$};
    \node[font=\scriptsize, above=1pt] at (R-1-\k.north) {$\mathcal{G}_{\k}$};}
\end{tikzpicture}
\vspace{0.5cm}
\\[6pt]
\textbf{HT} &
\begin{tikzpicture}[scale=0.55,>=latex, baseline={(current bounding box.center)}]
  \node[physnode, label={[corelbl]left:$\mathcal{G}_1$}] (G1) at (0,0) {};
  \node[physnode, label={[corelbl]left:$\mathcal{G}_2$}] (G2) at (1.5,0) {};
  \node[physnode, label={[corelbl]left:$\mathcal{G}_3$}] (G3) at (3,0) {};
  \node[physnode, label={[corelbl]left:$\mathcal{G}_4$}] (G4) at (4.5,0) {};
  \node[intnode,  label={[corelbl]left:$\mathcal{G}_5$}] (G5) at (0.75,1.3) {};
  \node[intnode,  label={[corelbl]right:$\mathcal{G}_6$}] (G6) at (3.75,1.3) {};
  \node[intnode,  label={[corelbl]above:$\mathcal{G}_7$}] (G7) at (2.25,2.6) {};
  \draw[tnbond] (G7)--(G5); \draw[tnbond] (G7)--(G6);
  \draw[tnbond] (G5)--(G1); \draw[tnbond] (G5)--(G2);
  \draw[tnbond] (G6)--(G3); \draw[tnbond] (G6)--(G4);
  \foreach \i in {1,...,4}{\draw[tnleg] (G\i) -- ++(0,-0.8) node[below,font=\scriptsize]{$i_{\i}$};}
\end{tikzpicture}
&
\begin{tikzpicture}[baseline={(R.center)}]
  \matrix (R) [matrix of math nodes, ampersand replacement=\&,
      left delimiter=(, right delimiter=),
      nodes={minimum size=0.5cm, inner sep=1pt, anchor=center, font=\scriptsize}]
  {
    n_1 \& 1 \& 1 \& 1 \& r_{1,5} \& 1 \& 1 \\
    1 \& n_2 \& 1 \& 1 \& r_{2,5} \& 1 \& 1 \\
    1 \& 1 \& n_3 \& 1 \& 1 \& r_{3,6} \& 1 \\
    1 \& 1 \& 1 \& n_4 \& 1 \& r_{4,6} \& 1 \\
    r_{1,5} \& r_{2,5} \& 1 \& 1 \& 1 \& 1 \& r_{5,7} \\
    1 \& 1 \& r_{3,6} \& r_{4,6} \& 1 \& 1 \& r_{6,7} \\
    1 \& 1 \& 1 \& 1 \& r_{5,7} \& r_{6,7} \& 1 \\
  };
  \foreach \k in {1,...,7}{
    \node[font=\scriptsize, left=10pt]  at (R-\k-1.west)  {$\mathcal{G}_{\k}$};
    \node[font=\scriptsize, above=1pt] at (R-1-\k.north) {$\mathcal{G}_{\k}$};}
\end{tikzpicture}
\\
\bottomrule
\end{tabular*}
\caption{Examples of well-known tensor formats expressed as graphs (left) and rank adjacency matrices (right). For FCTN, the graph is the complete graph $K_d$, with the corresponding matrix fully populated with ranks. For TT, the graph is the path graph $P_d$, with only the superdiagonal and subdiagonal entries populated with ranks. For HT, the graph is a balanced binary tree, requiring internal cores $\mathcal{G}_5, \mathcal{G}_6, \mathcal{G}_7$.}
    \label{fig:tn_rm}
\end{figure}

To describe the main idea underpinning our GTN-SVD algorithm, we begin with a simple observation: any simple graph on $p$ vertices is a subgraph of the complete graph $K_p$. The FCTN format is precisely the tensor network realization of $K_p$, where each of its $p$ cores is associated with an index and every pair of cores shares an edge. An arbitrary GTN can then be recovered from this complete network by trivializing its absent edges.

\subsection{GTN-SVD Algorithm}

To describe the main idea of our GTN-SVD algorithm, we begin with the observation that any simple graph on $p$ vertices is a subgraph of the complete graph $K_p$. The FCTN format is precisely the tensor network realization of $K_p$, where each of its $p$ cores is associated with an index and every pair of cores shares an edge. An arbitrary GTN can then be recovered from this complete network by {\em trivializing its absent edges}.
To this end, let $\mathbf{R} \in \mathbb{N}^{p \times p}$ be the rank adjacency matrix, with $\mathbf{R}_{jj} = n_j$ the size of the $j$-th physical mode and $\mathbf{R}_{ij} = r_{i,j}$, $i \neq j$, the size of the edge between cores $i$ and $j$. Since indexing over a singleton dimension acts as the identity, setting $\mathbf{R}_{ij} = 1$ deletes the edge $\{i,j\}$ and setting $n_j = 1$ removes the physical index of core $j$, neither of which alters the original tensor (see Figure~\ref{fig:tn_rm}).  
As a consequence, every graph tensor network on $p$ cores is isomorphic to an FCTN on $p$ cores whose corresponding entries of $\mathbf{R}$ equal one. Figure~\ref{fig:gtn_iso} illustrates this for a 4D tensor represented on $p = 6$ cores, two of which are internal ($n_5 = n_6 = 1$). The singleton ranks delete the missing edges, recovering an FCTN on six cores.
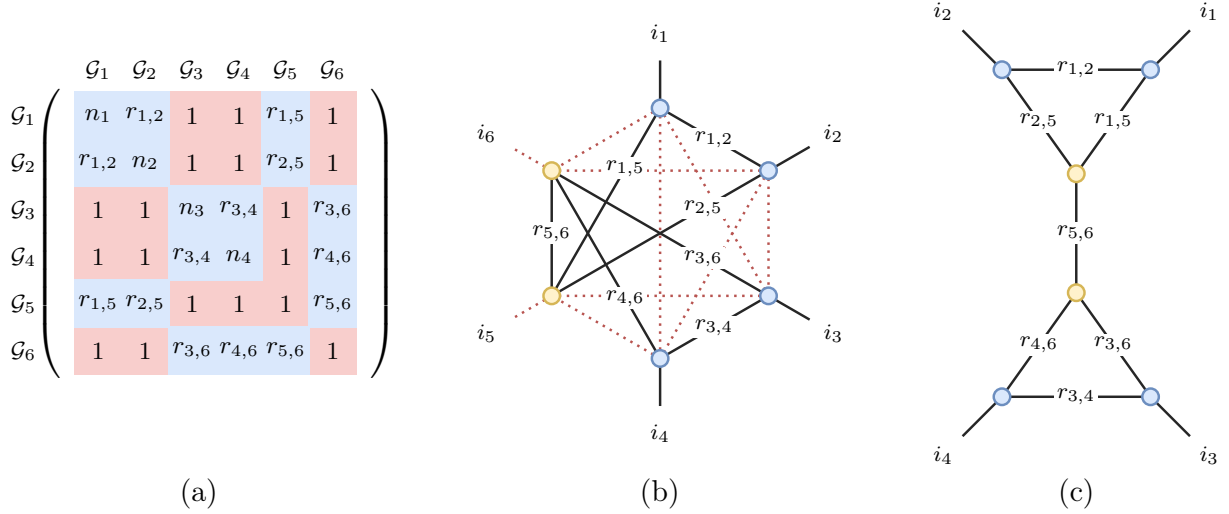
\begin{figure}[t]
  \centering
  \resizebox{1\textwidth}{!}{\tikzset{
  trivbond/.style={line width=0.9pt, edgered, dotted},
  bondlbl/.style={fill=white, inner sep=1pt, font=\scriptsize},
}
\begin{tabular}{@{}c@{\hspace{9mm}}c@{\hspace{9mm}}c@{}}
\begin{tikzpicture}[baseline={(R.center)}]
  \matrix (R) [matrix of math nodes, ampersand replacement=\&,
      left delimiter=(, right delimiter=),
      nodes={minimum size=0.6cm, inner sep=1pt, anchor=center, font=\footnotesize}]
  {
    n_1 \& r_{1,2} \& 1 \& 1 \& r_{1,5} \& 1 \\
    r_{1,2} \& n_2 \& 1 \& 1 \& r_{2,5} \& 1 \\
    1 \& 1 \& n_3 \& r_{3,4} \& 1 \& r_{3,6} \\
    1 \& 1 \& r_{3,4} \& n_4 \& 1 \& r_{4,6} \\
    r_{1,5} \& r_{2,5} \& 1 \& 1 \& 1 \& r_{5,6} \\
    1 \& 1 \& r_{3,6} \& r_{4,6} \& r_{5,6} \& 1 \\
  };
  \begin{scope}[on background layer]
    \foreach \i in {1,...,6}{\foreach \j in {1,...,6}{
      \fill[hlred] (R-\i-\j.north west) rectangle (R-\i-\j.south east);}}
    \foreach \i/\j in {1/1,1/2,1/5, 2/1,2/2,2/5, 3/3,3/4,3/6,
                       4/3,4/4,4/6, 5/1,5/2,5/6, 6/3,6/4,6/5}{
      \fill[hlblue] (R-\i-\j.north west) rectangle (R-\i-\j.south east);}
  \end{scope}
  \foreach \k in {1,...,6}{
    \node[font=\scriptsize, left=10pt]  at (R-\k-1.west)  {$\mathcal{G}_{\k}$};
    \node[font=\scriptsize, above=1pt] at (R-1-\k.north) {$\mathcal{G}_{\k}$};}
\end{tikzpicture}
&
\begin{tikzpicture}[scale=1.6,>=latex, baseline={(current bounding box.center)}]
  \def\rc{1}\def\rl{1.38}\def\rt{1.6}
  \foreach \i [evaluate=\i as \ang using 90-60*(\i-1)] in {1,...,6}{
    \coordinate (P\i) at (\ang:\rc);}
  \foreach \a/\b in {1/3,1/4,1/6,2/3,2/4,2/6,3/5,4/5}{\draw[trivbond] (P\a)--(P\b);}
  \foreach \a/\b in {1/2,1/5,2/5,3/4,3/6,4/6,5/6}{\draw[tnbond] (P\a)--(P\b);}
  \foreach \i [evaluate=\i as \ang using 90-60*(\i-1)] in {1,2,3,4}{
    \draw[tnleg] (P\i) -- (\ang:\rl);
    \node[font=\scriptsize] at (\ang:\rt) {$i_{\i}$};}
  \foreach \i [evaluate=\i as \ang using 90-60*(\i-1)] in {5,6}{
    \draw[trivbond] (P\i) -- (\ang:\rl);
    \node[font=\scriptsize] at (\ang:\rt) {$i_{\i}$};}
  \path (P1)--(P2) node[bondlbl,pos=0.50]{$r_{1,2}$};
  \path (P3)--(P4) node[bondlbl,pos=0.50]{$r_{3,4}$};
  \path (P5)--(P6) node[bondlbl,pos=0.50]{$r_{5,6}$};
  \path (P1)--(P5) node[bondlbl,pos=0.32]{$r_{1,5}$};
  \path (P4)--(P6) node[bondlbl,pos=0.32]{$r_{4,6}$};
  \path (P2)--(P5) node[bondlbl,pos=0.30]{$r_{2,5}$};
  \path (P3)--(P6) node[bondlbl,pos=0.30]{$r_{3,6}$};
  \foreach \i in {1,2,3,4}{\node[physnode] at (P\i) {};}
  \foreach \i in {5,6}{\node[intnode] at (P\i) {};}
\end{tikzpicture}
&
\begin{tikzpicture}[scale=0.95,>=latex, baseline={(current bounding box.center)}]
  \coordinate (N1) at ( 1, 2.2); \coordinate (N2) at (-1, 2.2);
  \coordinate (N5) at ( 0, 0.8); \coordinate (N6) at ( 0,-0.8);
  \coordinate (N3) at ( 1,-2.2); \coordinate (N4) at (-1,-2.2);
  \draw[tnbond] (N1)--(N2); \draw[tnbond] (N1)--(N5); \draw[tnbond] (N2)--(N5);
  \draw[tnbond] (N5)--(N6);
  \draw[tnbond] (N3)--(N6); \draw[tnbond] (N4)--(N6); \draw[tnbond] (N3)--(N4);
  \draw[tnleg] (N1) -- ++(45:0.75);  \node[font=\scriptsize] at ($(N1)+(45:1.12)$)  {$i_1$};
  \draw[tnleg] (N2) -- ++(135:0.75); \node[font=\scriptsize] at ($(N2)+(135:1.12)$) {$i_2$};
  \draw[tnleg] (N3) -- ++(315:0.75); \node[font=\scriptsize] at ($(N3)+(315:1.12)$) {$i_3$};
  \draw[tnleg] (N4) -- ++(225:0.75); \node[font=\scriptsize] at ($(N4)+(225:1.12)$) {$i_4$};
  \path (N1)--(N2) node[bondlbl,pos=0.5]{$r_{1,2}$};
  \path (N1)--(N5) node[bondlbl,pos=0.5]{$r_{1,5}$};
  \path (N2)--(N5) node[bondlbl,pos=0.5]{$r_{2,5}$};
  \path (N5)--(N6) node[bondlbl,pos=0.5]{$r_{5,6}$};
  \path (N3)--(N6) node[bondlbl,pos=0.5]{$r_{3,6}$};
  \path (N4)--(N6) node[bondlbl,pos=0.5]{$r_{4,6}$};
  \path (N3)--(N4) node[bondlbl,pos=0.5]{$r_{3,4}$};
  \node[physnode] at (N1) {}; \node[physnode] at (N2) {};
  \node[physnode] at (N3) {}; \node[physnode] at (N4) {};
  \node[intnode]  at (N5) {}; \node[intnode]  at (N6) {};
\end{tikzpicture}
\\[3mm]
{\normalsize (a)} & {\normalsize (b)} & {\normalsize (c)}
\end{tabular}}  
    \caption{(a) Rank adjacency matrix for a graph tensor network. Red entries equal to one mark either a disconnection between cores (off-diagonal) or an internal core with no physical index (diagonal). (b) Equivalent FCTN representation from rank adjacency matrix (a), in which the red dotted edges correspond to the unit entries and the blue solid lines correspond to a nontrivial edge rank or a physical index. (c) GTN  isomorphic to the FCTN depicted in (b).}
  \label{fig:gtn_iso}
\end{figure}
The isomorphism between GTN and FCTN allows us to extend the FCTN-SVD algorithm, originally proposed by Wang and Li \cite{wangFCTN2024}, to SVD-based decompositions of GTNs corresponding to any user-defined rank adjacency matrix $\mathbf{R}$. 
The new algorithm, which we call {GTN-SVD} (see Algorithm~\ref{alg:gtn-svd}), takes in a full tensor $\mathcal{X}$, a user-prescribed $p\times p$ rank adjacency matrix $\mathbf{R}$, and a prescribed error tolerance $\varepsilon$ to produce tensor cores $\{\mathcal{G}_1,\ldots,\mathcal{G}_p\}$ and, correspondingly, a GTN tensor $\tilde{\mathcal{X}}$ satisfying
\begin{equation}
\|\mathcal{X}-\tilde{\mathcal{X}}\|_F\le\varepsilon\,\|\mathcal{X}\|_F.
\label{epsilonX}
\end{equation}
The algorithm is based on the following error bound, 
which follows from the results in \cite{wangFCTN2024}.

\begin{theorem}[GTN-SVD error bound]
Let ${\mathcal X}\in \mathbb{R}^{n_1\times \cdots \times n_d}$ be a given 
tensor and $\tilde{\mathcal X}$ its GTN approximation using $p\geq d$ cores 
$\{\mathcal G_1,\ldots, \mathcal{G}_p\}$. Suppose that 
the $k$-unfoldings $\mathbf{X}_{<k>}$ of 
the tensor $\mathcal X$ satisfy
\begin{equation}
\mathbf{X}_{< k>}=\mathbf{A}_k+\mathbf{E}_k,\qquad 
\operatorname{rank}\left(\mathbf{A}_k\right) = r_k,\qquad
\|\mathbf{E}_k\|_F=\varepsilon_k,\qquad k=1,\dots,p-1,
\end{equation}
and that the outgoing edges of core $k$ are chosen so that
\begin{equation}
\prod_{\substack{j>k \\ \mathbf{R}_{kj}\neq 1}} \mathbf{R}_{kj} = r_k.
\end{equation}
Then GTN-SVD returns a tensor
$\tilde{\mathcal X}$ in the graph tensor network format prescribed by the rank
adjacency matrix $\mathbf{R}$, with
\begin{equation}
\big\|\mathcal X-\tilde{\mathcal X}\big\|_F
\le\sqrt{\sum_{k=1}^{p-1}\varepsilon_k^{2}}.
\label{p1}
\end{equation}
Moreover, if each SVD in Algorithm~\ref{alg:gtn-svd} is truncated at
\begin{equation}
\delta_k=\sqrt{\frac{2(p-k)}{p(p-1)}}\,\varepsilon\,\|\mathcal X\|_F, 
\label{delta}
\end{equation}
where $\varepsilon$ is a prescribed tolerance, then the error bound \eqref{epsilonX} holds.
\end{theorem}

\begin{proof}
The proof is by induction on $p$. For $p=2$ the statement follows from the properties of the SVD. Consider an arbitrary $p>2$. The first unfolding $\mathbf{X}_{<1>}$ is decomposed as
\[
\mathbf{X}_{<1>}=\mathbf{U}_1\mathbf{\Sigma}_1\mathbf{V}_1^{\top}+\mathbf{E}_1=\mathbf{U}_1\mathbf{B}_1+\mathbf{E}_1,
\]
\noindent
where $\mathbf{U}_1$ is of size $n_1\times r_1$, has orthonormal columns, and $\|\mathbf{E}_1\|_F=\varepsilon_1$. The first core $\mathcal G_1$ is obtained by reshaping the $r_1$ columns of $\mathbf{U}_1$ onto the present outgoing edges $\{j>1:R_{1j}\neq1\}$ of core $1$; as this re-indexes the columns, it leaves $\mathbf{U}_1$ unchanged as a linear map. The matrix $\mathbf{B}_1$ is associated with a $(p-1)$-core tensor $\mathcal B_1$, which GTN-SVD decomposes further; that is, $\mathbf{B}_1$ is approximated by some matrix $\widehat{\mathbf{B}}_1$. From the properties of the SVD it follows that $\mathbf{U}_1^{\top}\mathbf{E}_1=0$, and thus
\[
\|\mathcal X-\tilde{\mathcal X}\|_F^{2}
=\|\mathbf{X}_{<1>}-\mathbf{U}_1\widehat{\mathbf{B}}_1\|_F^{2}
=\|\mathbf{X}_{<1>}-\mathbf{U}_1\mathbf{B}_1\|_F^{2}+\|\mathbf{U}_1(\mathbf{B}_1-\widehat{\mathbf{B}}_1)\|_F^{2},
\]
and since $\mathbf{U}_1$ has orthonormal columns,
\[
\|\mathcal X-\tilde{\mathcal X}\|_F^{2}
\le\varepsilon_1^{2}+\|\mathbf{B}_1-\widehat{\mathbf{B}}_1\|_F^{2}.
\]
The matrix $\mathbf{B}_1$ is expressed from $\mathbf{X}_{\langle1\rangle}$ as $\mathbf{B}_1=\mathbf{U}_1^{\top}\mathbf{X}_{\langle1\rangle}$, and from the orthonormality of the columns of $\mathbf{U}_1$ it is not difficult to see that the distance of the $k$-th unfolding ($k=2,\dots,p-1$) of the $(p-1)$-core tensor $\mathcal B_1$ to the rank-$r_k$ matrix cannot be larger than $\varepsilon_k$. Proceeding by induction, 
\[
\|\mathbf{B}_1-\widehat{\mathbf{B}}_1\|_F^{2}\le\sum_{k=2}^{p-1}\varepsilon_k^{2},
\]
which together with the previous inequality gives \eqref{p1}. If each SVD is truncated at $\delta_k$ as in \eqref{delta}, then $\varepsilon_k\le\delta_k$ and
\[
\sum_{k=1}^{p-1}\delta_k^{2}
=\varepsilon^{2}\|\mathcal X\|_F^{2}\,\frac{2}{p(p-1)}\sum_{k=1}^{p-1}(p-k)
=\varepsilon^{2}\|\mathcal X\|_F^{2},
\]
so that $\|\mathcal X-\tilde{\mathcal X}\|_F\le\varepsilon\,\|\mathcal X\|_F$.

\end{proof}

\begin{algorithm}[t]
\caption{\textsc{GTN-SVD}: decomposing a full tensor to GTN cores}
\label{alg:gtn-svd}
\begin{algorithmic}[1]
\Require Dense tensor $\mathcal{X}\in\mathbb{R}^{n_1\times \ldots\times n_d}$; rank adjancency matrix $\mathbf{R}$; prescribed relative error $\varepsilon$
\Ensure  Tensor cores $\{\mathcal{G}_1,\ldots,\mathcal{G}_p \}$ of a GTN 
$\tilde{\mathcal{X}}$  such that $||\mathcal{X} - \tilde{\mathcal{X}}||_F\leq \varepsilon||\mathcal{X}||_F$

\State Compute truncation parameters for $k=1,\ldots,p-1$
\begin{equation*}
    \delta_k \gets \dfrac{\sqrt{2(d-k)}}{\sqrt{d(d-1)}}||\varepsilon \mathcal{X} ||_F
\end{equation*}

\State Compute low-rank approximation via $\delta_1$-truncated SVD: $\mathbf{X}_{<1>}=\mathbf{U}_1\mathbf{\Sigma}_1\mathbf{V}_1^T + \mathbf{E}_1$

\State Compute ranks $[r_{1,2}, r_{1,3},\ldots,r_{1,p}]$ via \textsc{SplitRankRule}~\cite{wangFCTN2024} using $\text{rank}(\mathbf{U_1})$ and $\mathbf{R}[1,:]$

\State $\mathcal{G}_1 \gets \operatorname{Reshape}(\mathbf{U}_1,[n_1,r_{1,2},r_{1,3},...,r_{1,p}])$

\State Temporary tensor $\mathcal{C} \gets \operatorname{Reshape} (\mathbf{\Sigma}_1\mathbf{V}_1^T, [r_{1,2},r_{1,3},\ldots,r_{1,p},n_2,n_3,\ldots,n_p])$

\State Permute $\mathcal{C}$ such that the new axis order is 
\begin{equation*}
    [r_{1,2}, n_2, r_{1,3}, n_3, \ldots, r_{1,p}, n_p]
\end{equation*}

\For{$k = 2$ to $p-1$}
    \State $\mathbf{C}\gets \operatorname{Reshape}(\mathcal{C}, [n_k\prod^{k-1}_ir_{i,k}, \quad (\prod^{p}_{j=k+1} n_j)(\prod^{k-1}_{i}\prod^{p}_{j=k+1} n_jr_{i,j})])$

    \State Compute low-rank approximation via $\delta_k$-truncated SVD: $\mathbf{C}_{<k>}=\mathbf{U}\mathbf{\Sigma}\mathbf{V}^T + \mathbf{E}_k$

    \State Compute ranks $[r_{k,k+1}, r_{k,k+2},\ldots,r_{k,p}]$ via \textsc{SplitRankRule}~\cite{wangFCTN2024} using $\text{rank}(\mathbf{U}_1)$ and $\mathbf{R}[k,:]$

    \State $\mathcal{G}_k \gets \operatorname{Reshape}(\mathbf{U}_k,[r_{1,k}, r_{2,k},\ldots,r_{k-1,k}, n_k,])$

    \State $\mathcal{C} \gets \operatorname{Reshape}\bigl(\mathbf{\Sigma}_k \mathbf{V}_k^{T},\;
           [\,r_{k,k+1}, r_{k,k+2}, \ldots, r_{k,p},
           \Big(\textstyle\prod_{i=1}^{k-1} r_{i,k+1}\Big), n_{k+1}, \ldots, \Big(\textstyle\prod_{i=1}^{k-1} r_{i,p}\Big), n_p\,]\bigr)$
    
    \State Permute $\mathcal{C}$ such that the new axis order is
    \begin{equation*}
        \Big[\;
        \Big(\textstyle\prod_{i=1}^{k-1} r_{i,k+1}\Big),\, r_{k,k+1},\, n_{k+1},\;
        \ldots,\;
        \Big(\textstyle\prod_{i=1}^{k-1} r_{i,p}\Big),\, r_{k,p},\, n_{p}
        \;\Big]
    \end{equation*}

\EndFor

\State $\mathcal{G}_p\gets\operatorname{Reshape}(\mathcal{C}, [r_{1,p}, r_{2,p}, \ldots, r_{p-1,p}, n_p])$

\end{algorithmic}
\end{algorithm}

Let us make a couple remarks to the proposed GTN-SVD Algorithm~\ref{alg:gtn-svd}. First we address the rank assignment strategy. We use the rank adjacency matrix $\mathbf{R}$ both to declare the network topology and to mark which edges the algorithm should assign ranks to. Our implementation uses $r_{k,j}=-1$ to indicate/initialize which edges are assigned ranks. The particular value $-1$ is immaterial: any marker that cannot be confused with an edge label works, that is, any indicator that is not $1$ or above. At each step the columns of the orthonormal factor are reshaped onto the edges carrying this indicator.
Next, we address the rank distribution strategy. When a core carries more than one edge flagged for assignment, the rank $r_k$ returned by the truncated SVD must be split among them. We follow the \emph{Split Rank Rule} (SRR) of Wang and Li~\cite{wangFCTN2024} (Algorithm~2 of FCTN-SVD), which factors $r_k$ into prime factors and distributes them across the flagged edges. We summarize SRR in Algorithm~\ref{alg:srr}. We emphasize that the SRR is a heuristic: it is not known to be an optimal way to allocate $r_k$, and other distributions of the rank among the edges may yield smaller representations or better accuracy. We simply adopt it for ease of implementation. This algorithm can be modified to prescribe ranks instead of assignment to adapt to a prescribed error tolerance, as long as each row $m=1,\ldots,p$ of the rank adjacency matrix suffices the following to ensure the truncation per iteration possible:
\begin{equation}
    \prod_{k=1}^{m}r_{m,k}\geq \prod_{l=m+1}^{p}r_{m,l}.
\end{equation}
\noindent
This constraint also leads us to believe the best organization of a rank adjacency matrix is to decompose all physical modes of the full tensor before decomposing the internal cores. Refer to the HT rank adjacency matrix in Figure~\ref{fig:tn_rm} for an example.

\begin{algorithm}[t]
\caption{\textsc{SplitRankRule} (SRR)~\cite{wangFCTN2024}: distribute a rank among outgoing bonds}
\label{alg:srr}
\begin{algorithmic}[1]
\Require Rank $r_k$ to be split; row $\mathbf{R}[k,:]$ identifying the $m$ outgoing bonds of core $k$ flagged for sizing
\Ensure Vector of split ranks $\mathbf{s} = [\,r_{k,k+1}, \ldots, r_{k,p}\,]$ for those bonds, with all other bonds set to $1$
\State Initialize $\mathbf{s} \gets \operatorname{ones}(1{:}m)$
\State Compute the prime factors of $r_k$ and sort them in descending order into a vector $\mathbf{a}$
\State $\mathrm{dim} \gets \operatorname{length}(\mathbf{a})$
\If{$\mathrm{dim} \leq m$}
    \State $\mathbf{s}(1{:}\mathrm{dim}) \gets \mathbf{a}$
\Else
    \State $\mathbf{s} \gets \mathbf{a}(1{:}m)$
    \State $\mathbf{a} \gets \operatorname{flip}\!\big(\mathbf{a}(m{+}1{:}\mathrm{dim})\big)$
    \State $\mathrm{dim} \gets \operatorname{length}(\mathbf{a})$
    \While{$\mathrm{dim} > m$}
        \State $\mathbf{s} \gets \mathbf{s} \odot \mathbf{a}(1{:}m)$
            \Comment{$\odot$ : Hadamard (elementwise) product}
        \State $\mathbf{a} \gets \mathbf{a}(m{+}1{:}\mathrm{dim})$
        \State $\mathrm{dim} \gets \operatorname{length}(\mathbf{a})$
    \EndWhile
    \State $\mathbf{s}(m{-}\mathrm{dim}{+}1{:}m) \gets \mathbf{s}(m{-}\mathrm{dim}{+}1{:}m) \odot \mathbf{a}$
\EndIf
\end{algorithmic}
\end{algorithm}

\section{Addition and Rounding Operations for Traceable GTNs}
\label{sec:operations}

In this section, we introduce three operations on GTN tensor formats, namely addition, Hadamard product, and rounding, which significantly broaden the applicability of the GTN format to high-dimensional problems. Before doing so, we briefly introduce the concept of {\em traceable graphs}, which is foundational to both the addition and rounding operations. We then define addition and the Hadamard product, which together allow us to define a norm operation. Finally, we propose a rounding procedure and discuss how to implement iterative rounding schemes in which both addition and rounding are used. Such iterative schemes arise naturally, for instance, when integrating high-dimensional linear PDEs on GTN tensor manifolds using step-truncation methods.

\subsection{Traceable Paths in Graph Tensor Networks}

SVD-based algorithms in the TT format owe much of their effectiveness to the underlying path structure: because the cores are arranged in a line, there is an unambiguous ordering in which to perform QR and SVD operations. This ordering enables well-defined operations such as orthogonalization and recompression, as required, for instance, in TT-rounding. A GTN format, on the other hand, is not guaranteed such an arrangement, since the presence of a cycle leaves the traversal order undetermined. This poses a challenge for defining operations such as rounding, and motivates the search for a traversal ordering of the tensor network in GTN format that visits each core exactly once.

This is a \emph{traceable path} (or Hamiltonian path): a path in the network that visits every node, without repeating any node or edge. A graph that admits such a path is called \emph{traceable}. In graph theory terms, a traceable path is a spanning subgraph of the network, so ordering the cores along it exposes a TT-like structure inside the GTN, as shown in Figure~\ref{fig:gtn_trace}. We want to note that deciding traceability is NP-complete in general, but the graphs arising in tensor decompositions are small---their size is governed by the dimensionality of the problem, not the mode sizes---so a path can be located by exhaustive or heuristic search at negligible cost relative 
to the decomposition itself.
\begin{figure}[t]
  \centering
  \resizebox{1\textwidth}{!}{\begin{tabular}{@{}c@{\hspace{14mm}}c@{}}
\begin{tikzpicture}[scale=1.05,>=latex, baseline={(current bounding box.center)}]
  \coordinate (N1) at (-2.2, 1); \coordinate (N2) at (-2.2,-1);
  \coordinate (N5) at (-0.8, 0); \coordinate (N6) at ( 0.8, 0);
  \coordinate (N3) at ( 2.2, 1); \coordinate (N4) at ( 2.2,-1);
    \draw[tnbond]    (N1)--(N2); \draw[tnbond] (N1)--(N5); \draw[tnbond] (N2)--(N5);
      \draw[tnbond]    (N5)--(N6);
      \draw[tnbond]    (N3)--(N6); \draw[tnbond] (N4)--(N6); \draw[tnbond] (N3)--(N4);
  \draw[tnleg] (N1) -- ++(135:0.7) node[font=\scriptsize, shift={(135:0.2)}]{$i_1$};
  \draw[tnleg] (N2) -- ++(225:0.7) node[font=\scriptsize, shift={(225:0.2)}]{$i_2$};
  \draw[tnleg] (N3) -- ++(45:0.7)  node[font=\scriptsize, shift={(45:0.2)}]{$i_3$};
  \draw[tnleg] (N4) -- ++(315:0.7) node[font=\scriptsize, shift={(315:0.2)}]{$i_4$};
  \draw[tracepath] (N1) to[bend right=22] (N2);
  \draw[tracepath] (N2) to[bend right=26] (N5);
  \draw[tracepath] (N5) to[bend right=30] (N6);
  \draw[tracepath] (N6) to[bend left=26]  (N3);
  \draw[tracepath] (N3) to[bend left=22]  (N4);
  \node[physnode, label={[corelbl]above:$\mathcal{G}_1$}] at (N1) {};
  \node[physnode, label={[corelbl]below:$\mathcal{G}_2$}] at (N2) {};
  \node[physnode, label={[corelbl]above:$\mathcal{G}_3$}] at (N3) {};
  \node[physnode, label={[corelbl]below:$\mathcal{G}_4$}] at (N4) {};
  \node[intnode,  label={[corelbl]above:$\mathcal{G}_5$}] at (N5) {};
  \node[intnode,  label={[corelbl]above:$\mathcal{G}_6$}] at (N6) {};
\end{tikzpicture}
&
\begin{tikzpicture}[scale=1.0,>=latex, baseline={(current bounding box.center)}]
  \coordinate (M1) at (0.00,0); \coordinate (M2) at (1.45,0);
  \coordinate (M5) at (2.90,0); \coordinate (M6) at (4.35,0);
  \coordinate (M3) at (5.80,0); \coordinate (M4) at (7.25,0);
  \draw[tnbond] (M1)--(M2); \draw[tnbond] (M2)--(M5); \draw[tnbond] (M5)--(M6);
  \draw[tnbond] (M6)--(M3); \draw[tnbond] (M3)--(M4);
  \draw[tnbond] (M1) to[bend left=55] (M5);
  \draw[tnbond] (M6) to[bend left=55] (M4);
  \draw[tracepath] (M1) to[bend right=36] (M2);
  \draw[tracepath] (M2) to[bend right=36] (M5);
  \draw[tracepath] (M5) to[bend right=36] (M6);
  \draw[tracepath] (M6) to[bend right=36] (M3);
  \draw[tracepath] (M3) to[bend right=36] (M4);
  \draw[tnleg] (M1) -- ++(270:0.9) node[below,font=\scriptsize]{$i_1$};
  \draw[tnleg] (M2) -- ++(270:0.9) node[below,font=\scriptsize]{$i_2$};
  \draw[tnleg] (M3) -- ++(270:0.9) node[below,font=\scriptsize]{$i_3$};
  \draw[tnleg] (M4) -- ++(270:0.9) node[below,font=\scriptsize]{$i_4$};
  \node[physnode, label={[corelbl]above left:$\mathcal{G}_1$}]  at (M1) {};
  \node[physnode, label={[corelbl]above:$\mathcal{G}_2$}]       at (M2) {};
  \node[intnode,  label={[corelbl]above right:$\mathcal{G}_5$}] at (M5) {};
  \node[intnode,  label={[corelbl]above left:$\mathcal{G}_6$}]  at (M6) {};
  \node[physnode, label={[corelbl]above:$\mathcal{G}_3$}]       at (M3) {};
  \node[physnode, label={[corelbl]above right:$\mathcal{G}_4$}] at (M4) {};
\end{tikzpicture}
\\[3mm]
{\normalsize (a)} & {\normalsize (b)}
\end{tabular}}  
\caption{(a) A traceable path through an arbitrary GTN: the dotted blue arrows visit every core once, in the order
$\mathcal{G}_1\!\to\!\mathcal{G}_2\!\to\!\mathcal{G}_5\!\to\!\mathcal{G}_6\!\to\!\mathcal{G}_3\!\to\!\mathcal{G}_4$, crossing the edges $r_{1,2}\!\to\!r_{2,5}\!\to\!r_{5,6}\!\to\!r_{3,6}\!\to\!r_{3,4}$ between consecutive cores.
(b) Re-indexing the cores in this path order reveals the TT-like substructure.}
  \label{fig:gtn_trace}
\end{figure}
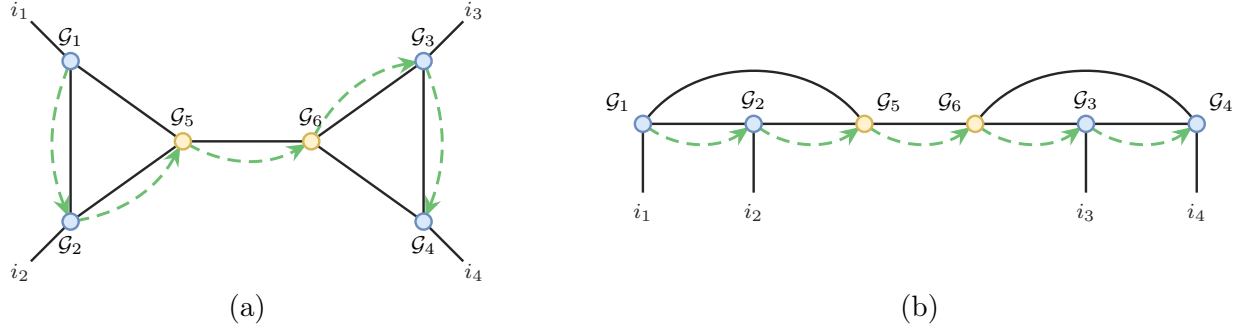
\begin{figure}[t]
    \centering
    \resizebox{\textwidth}{!}{\begin{tikzpicture}[scale=1,>=latex,
    yes/.style={font=\small\bfseries, edgegreen},
    no/.style={font=\small\bfseries, xred}]
\begin{scope}[shift={(0,0)}]
    \foreach \i in {1,...,4} {\coordinate (TT\i) at ({\i-2.5}, 0.3);}
    \draw[tnbond] (TT1)--(TT2)--(TT3)--(TT4);                 
    \foreach \i in {1,...,4} {                                
        \draw[tnleg] (TT\i) -- ++(0,-0.6) node[below,font=\scriptsize] {$i_{\i}$};}
    \draw[tracepath] (TT1) to[bend left=50] (TT2);
    \draw[tracepath] (TT2) to[bend left=50] (TT3);
    \draw[tracepath] (TT3) to[bend left=50] (TT4);
    \foreach \i in {1,...,4} {\node[physnode] at (TT\i) {};} 
    \node[yes] at (0,-2.3) {traceable};
\end{scope}
\begin{scope}[shift={(4,0)}]
    \def\rc{0.8}\def\rl{1.4}\def\rt{1.75}
    \foreach \i in {1,...,4} {\coordinate (TR\i) at ({45+90*(\i-1)}:\rc);}
    \draw[tnbond] (TR1)--(TR2)--(TR3)--(TR4)--(TR1);          
    \foreach \i in {1,...,4} {                                
        \draw[tnleg] (TR\i) -- ({45+90*(\i-1)}:\rl);
        \node[font=\scriptsize] at ({45+90*(\i-1)}:\rt) {$i_{\i}$};}
    \draw[tracepath] (TR1) to[bend right=55] (TR2);
    \draw[tracepath] (TR2) to[bend right=55] (TR3);
    \draw[tracepath] (TR3) to[bend right=55] (TR4);
    \foreach \i in {1,...,4} {\node[physnode] at (TR\i) {};}
    \node[yes] at (0,-2.3) {traceable};
\end{scope}
\begin{scope}[shift={(8,0)}]
    \def\rc{0.85}\def\rl{1.5}\def\rt{1.85}
    \foreach \i in {1,...,4} {\coordinate (F\i) at ({45+90*(\i-1)}:\rc);}
    \draw[tnbond] (F1)--(F2)--(F3)--(F4)--(F1);               
    \draw[tnbond] (F1)--(F3); \draw[tnbond] (F2)--(F4);       
    \foreach \i in {1,...,4} {                                
        \draw[tnleg] (F\i) -- ({45+90*(\i-1)}:\rl);
        \node[font=\scriptsize] at ({45+90*(\i-1)}:\rt) {$i_{\i}$};}
    \draw[tracepath] (F1) to[bend right=55] (F2);
    \draw[tracepath] (F2) to[bend right=55] (F3);
    \draw[tracepath] (F3) to[bend right=55] (F4);
    \foreach \i in {1,...,4} {\node[physnode] at (F\i) {};}
    \node[yes] at (0,-2.3) {traceable};
\end{scope}
\begin{scope}[shift={(12,0)}]
    \coordinate (H1) at (-1.5,-0.6);                  
    \coordinate (H2) at (-0.5,-0.6);
    \coordinate (H3) at ( 0.5,-0.6);
    \coordinate (H4) at ( 1.5,-0.6);
    \coordinate (H5) at (-1, 0.4);                    
    \coordinate (H6) at ( 1, 0.4);
    \coordinate (H7) at ( 0, 1.4);                    
    \draw[tnbond] (H7)--(H5); \draw[tnbond] (H7)--(H6);   
    \draw[tnbond] (H5)--(H1); \draw[tnbond] (H5)--(H2);
    \draw[tnbond] (H6)--(H3); \draw[tnbond] (H6)--(H4);
    \foreach \i in {1,...,4} {                        
        \draw[tnleg] (H\i) -- ++(0,-0.6) node[below,font=\scriptsize] {$i_{\i}$};}
    \draw[tracepath] (H1) to[bend left=45] (H5);
    \draw[tracepath] (H5) to[bend left=45] (H7);
    \draw[tracepath] (H7) to[bend left=45] (H6);
    \draw[tracepath] (H6) to[bend right=45] (H3);
    \draw[deadpath] (H5) to[bend left=45] (H2);
    \draw[deadpath] (H6) to[bend left=45] (H4);
    \foreach \i in {1,...,4} {\node[physnode] at (H\i) {};}   
    \foreach \i in {5,6,7} {\node[intnode] at (H\i) {};}      
    \node[no] at (0,-2.3) {not traceable};
\end{scope}
\end{tikzpicture}}
\caption{Traceability of common tensor network formats for $d=4$ (left to
right): TT, TR, FCTN, and HT. The dashed green arrows mark a
traceable path (except for HT) visiting every core once. TT is
trivial, being a path itself; TR becomes traceable by omitting a
single ring edge; and FCTN admits many traceable paths, one of which is shown. The HT tree admits none: two nodes cannot be visited without revisiting a node or edge.}
    \label{fig:gtn_traceable_ex}
\end{figure}
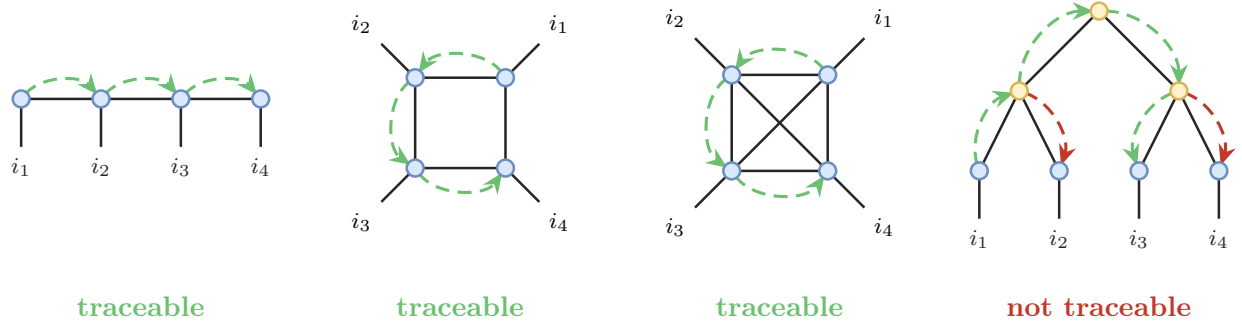
With a traceable path in hand, GTN is ``close enough'' to a TT that we can reuse TT-based algorithmic machinery. In particular, addition concatenates cores along the edges of the network, and rounding sweeps the traceable order to compress the result back to a prescribed tolerance.

\subsection{Tensor Addition}
\label{sec:add}
As is well known, for the TT and HT formats, addition is carried out by block concatenation: the cores of the two summands are stacked into disjoint blocks and padded with zeros off the diagonal, so that each edge dimension becomes the sum of the two incoming dimensions. This zero padding is a conceptual device rather than a computational one: in practice, neither the zero blocks are stored nor are any operations performed on them, so the added cost of addition is negligible relative to the size of the summands. Structurally, the zero padding is what makes the construction work: it guarantees that no contraction through the network can mix a core belonging to the first summand with a core belonging to the second. Every surviving term draws all of its cores from one summand or all from the other, so the network evaluates to $\mathcal{X}^1 + \mathcal{X}^2$.

Nothing about this argument fails when the network contains cycles. Concatenating every edge of a cyclic network block-diagonally still yields the exact sum, for the same reason. What fails is the aftermath: addition doubles every edge it touches, so in any iterative scheme (where a sum is formed at every step), the ranks must be compressed back down afterwards, and this is where the topology intervenes. Rounding a GTN proceeds by sweeping along a traceable path, and the edges lying on that path are orthogonalized and truncated; the remaining edges, i.e., the \emph{chords}, are carried along. A chord that addition has inflated from $r$ to $2r$ is therefore inflated: the sweep never touches it, and no subsequent rounding can undo the growth \cite{mickelin2020}. Over many steps the chord ranks increase upward without bound, even though the tensor they represent may not be dramatically changed than when it started.

This is resolved, following Mickelin and Karaman~\cite{mickelin2020}, by treating certain edges differently. Once the cores are laid out along a traceable path, each core's edges fall into one of two classes: the edges of the traceable path, joining the core to its immediate predecessor and successor, and the chords, joining it to other cores elsewhere in the network. Addition then concatenates along the traceable path, as in the tensor train, but overlays along the chords. Hence, after the addition, a traceable path edge of dimensions $r'$ and $r''$ becomes an edge of dimension $r' + r''$, while a chord of dimensions $r'$ and $r''$ becomes an edge of dimension $\max(r', r'')$, with both summands written into the chord starting from the same index so that they overlap rather than occupying separate blocks.
Writing $\alpha_{k-1}, \alpha_k$ for the traceable path indices of core $k$ and $\gamma$ for its chord indices, the cores of the sum are defined as
\begin{equation}
\mathcal{G}_k(\alpha_{k-1}, i_k, \alpha_k, \gamma) =
\begin{cases}
\mathcal{G}^{1}_k(\alpha_{k-1}, i_k, \alpha_k, \gamma),
& \alpha_{k-1}, \alpha_k \text{ in the first block},\ \gamma \le r',\\[4pt]
\mathcal{G}^{2}_k(\alpha_{k-1}, i_k, \alpha_k, \gamma),
& \alpha_{k-1}, \alpha_k \text{ in the second block},\ \gamma \le r'',\\[4pt]
0, & \text{otherwise,}
\end{cases}
\label{eq:gtn-add}
\end{equation}
The construction therefore inflates where inflation can be repaired. Traceable path edges grow additively, but they lie on the path, so the rounding procedure compresses them back to whatever the prescribed tolerance permits. Chords are frozen at the larger of the two incoming dimensions and never grow at all. This is not a concern in practice, since chord dimensions are governed by the network topology rather than by the number of additions performed, and remain small relative to the path ranks throughout. Repeated addition, the operation on which every step-truncation scheme rests, is in this way kept stable on networks with cycles.

\subsection{Hadamard Product and Frobenius Norm}
The Hadamard product between two tensors $\mathcal{X}^1$ and $\mathcal{X}^2$, denoted as $\mathcal{X}^1\odot \mathcal{X}^2$, is well suited to the graph format because it acts on each tensor entry independently. Fix an index tuple $(i_1,\ldots,i_p)$. The corresponding entry of $\mathcal{X}^1$ is a scalar obtained by contracting the network of slices 
$\{\mathcal{G}^1_1(i_1),\ldots,\mathcal{G}^1_p(i_p)\}$, and likewise for $\mathcal{X}^2$. The entry of the product is just the product of these two scalars. Two independent contractions may always be run side by side as a single contraction over the pair of networks, and the object that does this is the Kronecker product. Slicing each core at a fixed physical index and taking the product of the slices,
\begin{equation}
    \mathcal{G}_k(i_k)=\mathcal{G}^1_{k}(i_k)\otimes\mathcal{G}^2_{k}(i_k),
    \qquad i_k=1,\dots,n_k,\quad k=1,\ldots, p, 
    \label{eq:gtn-hadamard}
\end{equation}
\noindent
therefore produces cores whose contraction reproduces $\mathcal{X}^1 \odot \mathcal{X}^2$ entry by entry.
The graph topology is untouched by this construction, but the edge dimensions change. An edge that carried dimensions $r'_{i,j}$ and $r''_{i,j}$ in the two factors carries dimension
\begin{equation}
    r_{i,j}=r'_{i,j}\,r''_{i,j},
    \qquad i \ne j,
    \label{eq:gtn-hadamard-ranks}
\end{equation}
in the product, while the physical mode sizes $n_k$ are unchanged.
Everything else follows. The inner product of two tensors in the format is the sum of all entries of their Hadamard product, and summing all entries of a graph tensor network requires no reconstruction either. Taking $\mathcal{X}^1=\mathcal{X}^2=\mathcal{X}$ then gives the squared Frobenius norm of $\mathcal{X}$
\begin{equation}
    \|\mathcal{X}\|_F^2
    = \langle\mathcal{X},\mathcal{X}\rangle
    = \sum_{i_1,\dots,i_p}\mathcal{X}(i_1,\dots,i_p)^2
    = \sum_{i_1,\dots,i_p}(\mathcal{X}\odot\mathcal{X})(i_1,\dots,i_p),
    \label{eq:gtn-norm}
\end{equation}
where $(\mathcal{X}\odot\mathcal{X})(i_1,\dots,i_p)$ is computed directly from the cores.

\subsection{Tensor Rounding}

Tensor rounding (or truncation) is a crucial tensor operation, which takes a decomposed tensor with suboptimal ranks and returns one with compressed ranks. Formats like TT have a well-defined SVD-based rounding operation that is quasi-optimal \cite{oseledets2011}. Although rounding tensor formats with loops is known to be a challenging problem with fundamental limitations \cite{batselier2018}, here we propose a heuristic rounding operation that reduces the ranks of GTN formats containing loops with bounded error. To achieve this, the GTN must admit a traceable path.

Ordering the cores along a traceable path organizes the network into a chain with chords: along a chain a TT-like rounding procedure applies. We therefore round GTNs in two sweeps. The first passes right to left along the traceable path, replacing each core by an orthogonal factor and pushing the remaining triangular factor into its predecessor. This sweep is a change of basis and nothing more: it introduces no error, and it leaves the network orthogonalized. The second sweep passes left to right, taking a truncated SVD at each edge of the traceable path and discarding every singular value below a threshold $\delta$. Because the first sweep has orthogonalized everything to the right of the working core, the error committed at each edge is the tail of the discarded singular values. The chords take no part in either sweep. They are carried as additional legs of the cores they attach to, consistent with how chords are treated in the tensor addition operation we defined in Section \ref{sec:add}. 

On a network with no chords, i.e., a path, the proposed rounding procedure reduces to TT rounding and inherits its quasi-optimality. In the presence of chords it does not, and we claim no optimality. What it does retain, though, is an error bound. That bound rests on the choice of the truncation threshold $\delta$. 
Specifically, distributing a relative tolerance $\varepsilon$ across $m$ SVD truncations whose errors are governed by $\delta = \varepsilon \|\mathcal{X}\|_F / \sqrt{m}$ gives the desired bound. For instance, taking $m = d-1$ would recover the usual TT rule. Setting $m = (p-1) + c$, where $c$ is the number of chords, yields instead
\begin{equation}
	\delta  \;=\; \frac{\varepsilon}{\sqrt{(p-1)+c}}\,\|\tilde{\mathcal{X}}\|_F
	 \label{eq:gtn-delta}
\end{equation}
which budgets for truncations at the chords as well, even though the sweep performs none there. This is deliberately conservative: the quasi-optimality argument breaks down once the cores carry chord legs, so in the absence of a sharper bound, we compensate by shrinking every individual threshold. The cost is a mild over-truncation (ranks slightly larger than a sharper rule would permit) and the benefit is that the accumulated error remains below $\varepsilon\|\mathcal{X}\|_F$ in practice, uniformly across the topologies we tested.
The norm appearing in \eqref{eq:gtn-delta} is itself computed in the format by applying the Hadamard product discussed previously, so setting the truncation threshold never requires forming the full tensor.

\begin{algorithm}[t!]
\caption{\textsc{GTN-Round}: rounding a traceable path-ordered GTN}
\label{alg:gtn-round}
\begin{algorithmic}[1]
\Require Trace-ordered cores $\{\mathcal{G}_k\}_{k=1}^{p}$ with
         $\mathcal{G}_k\in\mathbb{R}^{\,r_{1,k}\times\cdots\times r_{k-1,k}\times n_k\times r_{k,k+1}\times\cdots\times r_{k,p}}$
         (axis $j$ is bond $(k,j)$ of size $r_{k,j}$, axis $k$ is the physical mode $n_k$);
         trace-ordered rank adjacency matrix $\mathbf{R}\in\mathbb{Z}^{p\times p}$ with
         $(\mathbf{R})_{i,j}=r_{i,j}$ and $(\mathbf{R})_{k,k}=n_k$;
         prescribed relative error $\varepsilon$.
\Ensure  Rounded cores $\{\mathcal{G}_k\}_{k=1}^{p}$ with
         $\|\mathcal{X}-\tilde{\mathcal{X}}\|_F \le \varepsilon\,\|\mathcal{X}\|_F$.
\Statex
\State Find number of skipped bonds $c$
\State Compute the truncation parameter
\Comment{Use GTN norm to compute $||\mathcal{\tilde{\mathcal{X}}}||_F$}
\begin{equation*}
    \delta \gets \dfrac{\varepsilon}{\sqrt{(p-1)+c}}\,\|\tilde{\mathcal{X}}\|_F
\end{equation*}
\Statex
\For{$k = p,\, p-1,\, \dots,\, 2$}
\Comment{Right-to-left orthogonalization sweep}
    \State Reshape $\mathcal{G}_k$ via its $(k\!-\!1)$-unfolding
           $\mathbf{G}^{(k)}_{(k-1)}\in\mathbb{R}^{\,r_{k-1,k}\times t_k}$,\ \ where
           $t_k=\textstyle\prod_{j=1,\,j\neq k-1}^{p} r_{j,k}$
    \State $[\mathbf{Q},\mathbf{T}_k]\gets\mathrm{QR}\!\big((\mathbf{G}^{(k)}_{(k-1)})^{T}\big)$,\ \ 
           $\mathbf{Q}\in\mathbb{R}^{\,t_k\times r'_{k-1,k}}$,\ 
           $\mathbf{T}_k\in\mathbb{R}^{\,r'_{k-1,k}\times r_{k-1,k}}$,\ 
           $r'_{k-1,k}=\min(t_k,r_{k-1,k})$
    \State Update $\mathbf{R}$: $r_{k-1,k}\gets r'_{k-1,k}$ \ 
    \State $\mathcal{G}_k \gets \operatorname{Reshape}\!\big(\mathbf{Q}^{T},\,
           [\,r_{k-1,k},\, r_{1,k},\ldots,r_{k-2,k},\, n_k,\, r_{k,k+1},\ldots,r_{k,p}\,]\big)$
    \State Permute $\mathcal{G}_k$ to axis order
    \begin{equation*}
        [\,r_{1,k},\ldots,r_{k-2,k},\, r_{k-1,k},\, n_k,\, r_{k,k+1},\ldots,r_{k,p}\,]
    \end{equation*}
    \State $\mathcal{G}_{k-1}\gets\mathcal{G}_{k-1}\times_{k}\mathbf{T}_k$
\EndFor
\Statex
\For{$k = 1,\, \dots,\, p-1$}
\Comment{Left-to-right compression sweep}
    \State Reshape $\mathcal{G}_k$ via its $(k\!+\!1)$-unfolding
           $\mathbf{G}^{(k)}_{(k+1)}\in\mathbb{R}^{\,r_{k,k+1}\times s_k}$,\ \ where
           $s_k=\textstyle\prod_{j=1,\,j\neq k+1}^{p} r_{j,k}$
    \State Compute low-rank approximation via $\delta$-truncated SVD:
           $(\mathbf{G}^{(k)}_{(k+1)})^{T}=\mathbf{U}\mathbf{\Sigma}\mathbf{V}^{T}+\mathbf{E}_k$,\ \ 
           $\mathbf{U}\in\mathbb{R}^{\,s_k\times r'_{k,k+1}}$,\ 
           $\mathbf{\Sigma}\mathbf{V}^{T}\in\mathbb{R}^{\,r'_{k,k+1}\times r_{k,k+1}}$
    \State Update $\mathbf{R}$: $r_{k,k+1}\gets r'_{k,k+1}$ \ 
    \State $\mathcal{G}_k \gets \operatorname{Reshape}\!\big(\mathbf{U},\,
           [\,r_{1,k},\ldots,r_{k-1,k},\, n_k,\, r_{k,k+2},\ldots,r_{k,p},\, r_{k,k+1}\,]\big)$
    \State Permute $\mathcal{G}_k$ to axis order
    \begin{equation*}
        [\,r_{1,k},\ldots,r_{k-1,k},\, n_k,\, r_{k,k+1},\, r_{k,k+2},\ldots,r_{k,p}\,]
    \end{equation*}
    \State $\mathcal{G}_{k+1}\gets\mathcal{G}_{k+1}\times_{k}(\mathbf{\Sigma}\mathbf{V}^{T})$
\EndFor
\end{algorithmic}
\end{algorithm}

\section{Numerical Experiments}
\label{sec:numerical_experiments} 
In this section we present numerical experiments demonstrating the efficiency of both GTN-SVD and GTN-rounding compared to their TT and HT counterparts. Throughout each experiment, we compare our Python implementation with the TT and HT tensor algorithm implemented in \verb|tt-toolbox| developed by Oseledets \cite{tttoolbox} and \verb|htucker| developed by Kressner \& Tobler \cite{htucker}.

Before presenting our results, we would like to provide a few comments on the formatting of the baseline TT and HT representations used in our comparisons. The ranks of a TT decomposition (and therefore its number of degrees of freedom) depend on the ordering of the tensor cores: permuting the cores can raise or lower the TT-ranks \cite{li2022}. This has been studied, for instance through algorithms that optimize the ordering to reduce the ranks \cite{tichavsky2025}. The same structural sensitivity motivates the active field of \emph{tensor network structure search} (TN-SS), which seeks the network topology best suited to a given task, including the minimization of storage or compression cost \cite{li2020, li2023, zeng2024}. We emphasize that our baselines implement none of these methods: we use the \verb|tt-toolbox| \cite{tttoolbox} and \verb|htucker| \cite{htucker} toolboxes with their default mode orderings and dimension trees, without any rearrangement or structure optimization. This is a deliberate choice, as no comparable arrangement-optimization techniques currently exist for the proposed GTN format. Structure optimization could improve the baselines and, in principle, our method, but this is left to future work.

\subsection{GTN-SVD of High-Dimensional Functions}

To assess the effectiveness of the GTN-SVD algorithm, we compare the TT and HT formats against two GTN topologies depicted in Figure~\ref{fig:gtn_bb_bt}. The first, which we call the \emph{barbell} (GTN-BB), is an eight-core graph in which the six ``physical'' cores (i.e., the cores with a free edge) are supplemented by two internal cores $\mathcal{G}_6,\mathcal{G}_7$, forming two $4$-cycles on $\{0,2,4,6\}$ and $\{1,3,5,7\}$ joined by the internal bridge $(6,7)$ (see Figure~\ref{fig:gtn_bb_bt}(a)). The second, the \emph{bowtie} (GTN-BT), uses no internal cores at all: its six physical cores form two $3$-cycles $\{0,2,4\}$ and $\{1,3,5\}$ joined by the bridge $(4,5)$ (see Figure~\ref{fig:gtn_bb_bt}(b)).
\begin{figure}[t]
  \centering
  \begin{tabular}{@{}c@{\hspace{12mm}}c@{}}
	\begin{tikzpicture}[scale=0.8, baseline={(current bounding box.center)}, >=latex]
		\coordinate (C1) at (-2.5,  1.5); \coordinate (C3) at (-4.0, 0);
		\coordinate (C5) at (-2.5, -1.5); \coordinate (C7) at (-1.0, 0);
		\coordinate (C8) at ( 1.0, 0); \coordinate (C2) at ( 2.5,  1.5);
		\coordinate (C4) at ( 4.0, 0); \coordinate (C6) at ( 2.5, -1.5);
		\draw[tnbond] (C1)--(C3); \draw[tnbond] (C3)--(C5);
		\draw[tnbond] (C5)--(C7); \draw[tnbond] (C7)--(C1);
		\draw[tnbond] (C2)--(C4); \draw[tnbond] (C4)--(C6);
		\draw[tnbond] (C6)--(C8); \draw[tnbond] (C8)--(C2);
		\draw[tnbond] (C7)--(C8);
		\draw[tnleg] (C1) -- ++(90:0.85)  node[above, font=\scriptsize]{$i_1$};
		\draw[tnleg] (C3) -- ++(180:0.85) node[left,  font=\scriptsize]{$i_3$};
		\draw[tnleg] (C5) -- ++(270:0.85) node[below, font=\scriptsize]{$i_5$};
		\draw[tnleg] (C2) -- ++(90:0.85)  node[above, font=\scriptsize]{$i_2$};
		\draw[tnleg] (C4) -- ++(0:0.85)   node[right, font=\scriptsize]{$i_4$};
		\draw[tnleg] (C6) -- ++(270:0.85) node[below, font=\scriptsize]{$i_6$};
		\node[physnode, label={[corelbl]left:$\mathcal{G}_1$}]  at (C1) {};
		\node[physnode, label={[corelbl]above:$\mathcal{G}_3$}] at (C3) {};
		\node[physnode, label={[corelbl]left:$\mathcal{G}_5$}]  at (C5) {};
		\node[intnode,  label={[corelbl]above:$\mathcal{G}_7$}] at (C7) {};
		\node[intnode,  label={[corelbl]above:$\mathcal{G}_8$}] at (C8) {};
		\node[physnode, label={[corelbl]right:$\mathcal{G}_2$}] at (C2) {};
		\node[physnode, label={[corelbl]above:$\mathcal{G}_4$}] at (C4) {};
		\node[physnode, label={[corelbl]right:$\mathcal{G}_6$}] at (C6) {};
	\end{tikzpicture}
	&
	\begin{tikzpicture}[scale=0.9, baseline={(current bounding box.center)}, >=latex]
		\coordinate (D1) at (-2.2, 1); \coordinate (D3) at (-2.2,-1);
		\coordinate (D5) at (-0.8, 0); \coordinate (D6) at ( 0.8, 0);
		\coordinate (D2) at ( 2.2, 1); \coordinate (D4) at ( 2.2,-1);
		\draw[tnbond] (D1)--(D3); \draw[tnbond] (D1)--(D5); \draw[tnbond] (D3)--(D5);
		\draw[tnbond] (D5)--(D6);
		\draw[tnbond] (D2)--(D6); \draw[tnbond] (D4)--(D6); \draw[tnbond] (D2)--(D4);
		\draw[tnleg] (D1) -- ++(135:0.85) node[above left, font=\scriptsize]{$i_1$};
		\draw[tnleg] (D3) -- ++(225:0.85) node[below left, font=\scriptsize]{$i_3$};
		\draw[tnleg] (D5) -- ++(270:0.85) node[below,      font=\scriptsize]{$i_5$};
		\draw[tnleg] (D6) -- ++(270:0.85) node[below,      font=\scriptsize]{$i_6$};
		\draw[tnleg] (D2) -- ++(45:0.85)  node[above right,font=\scriptsize]{$i_2$};
		\draw[tnleg] (D4) -- ++(315:0.85) node[below right,font=\scriptsize]{$i_4$};
		\node[physnode, label={[corelbl]left:$\mathcal{G}_1$}]  at (D1) {};
		\node[physnode, label={[corelbl]left:$\mathcal{G}_3$}]  at (D3) {};
		\node[physnode, label={[corelbl]above:$\mathcal{G}_5$}] at (D5) {};
		\node[physnode, label={[corelbl]above:$\mathcal{G}_6$}] at (D6) {};
		\node[physnode, label={[corelbl]right:$\mathcal{G}_2$}] at (D2) {};
		\node[physnode, label={[corelbl]right:$\mathcal{G}_4$}] at (D4) {};
	\end{tikzpicture}
	\\[8pt]
	{\normalsize (a)} & {\normalsize (b)}
\end{tabular}  
  \caption{Two GTN topologies representing a 6D tensor: (a) {\em Barbell} GTN format (GTN-BB); (b) {\em Bowtie} GTN format (GTN-BT).}
  \label{fig:gtn_bb_bt}
\end{figure}
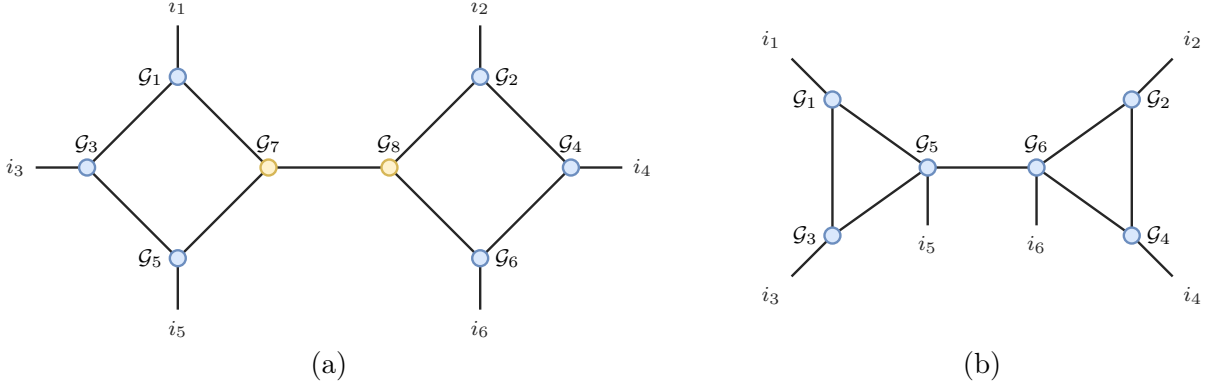

To demonstrate GTN-SVD, we first build two 6D dense tensors $\mathcal{X}_i$ ($i=1,2$) by evaluating the scalar functions 
\begin{equation*}
f_i:[0,2\pi]^6 \to\mathbb{R}
\end{equation*}
\begin{equation}
\begin{aligned}
f_1(\mathbf{x}) \;=\; \exp\Big(
&-3\big(\cos(x_0-\pi)+\cos(x_2-\pi)+\cos(x_4-\pi)-0.5\big)^2\\
&+ \sin(x_0-\pi)\sin(x_2-\pi) + \sin(x_2-\pi)\sin(x_4-\pi) + \sin(x_0-\pi)\sin(x_4-\pi)\\
&+ 1.2\big(\cos(x_1-x_3)+\cos(x_3-x_5)+\cos(x_1-x_5)\big)
+ 0.8\cos\!\big(2(x_1+x_3+x_5)\big)\Big),
\end{aligned}
\label{eq:comp1}
\end{equation}
and
\begin{equation}
\begin{aligned}
f_2(\mathbf{x}) \;=\; \Big[1 + \exp\Big(
&\;3\big(\cos(x_0-\pi)+\cos(x_2-\pi)+\cos(x_4-\pi)-0.5\big)^2\\
&- \sin(x_0-\pi)\sin(x_2-\pi) - \sin(x_2-\pi)\sin(x_4-\pi) - \sin(x_0-\pi)\sin(x_4-\pi)\\
&+ 1.2\big(\cos(x_1-x_3)+\cos(x_3-x_5)+\cos(x_1-x_5)\big)\\
&+ 0.8\cos\!\big(2(x_1+x_3+x_5)\big) - 2\Big)\Big]^{-1}
\end{aligned}
\label{eq:comp2}
\end{equation}
on a uniform Cartesian grid with $N=24$ points along each axis. This gives two tensors with $24^{6} \approx 1.9\times 10^{8}$ entries each. We then decompose each $\mathcal{X}_i$ using the GTN-SVD algorithm and the BB and BT topologies shown in Figure~\ref{fig:gtn_bb_bt}, and compare against the TT and HT formats. In particular, we record the total number degrees of freedom (DoF), the wall-clock time, and the relative $L^2$ error for each prescribed tolerance $\varepsilon \in \{10^{-4}, 10^{-6}, 10^{-8}, 10^{-10}\}$.
\begin{figure}[t]
    \centering
    \includegraphics[width=\linewidth]{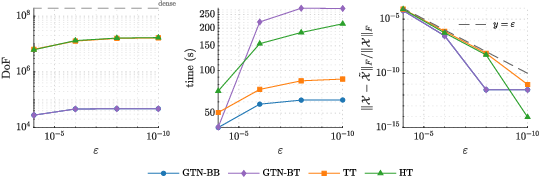}
\caption{Tensor decomposition of the  six-dimensional function \eqref{eq:comp1} using GTN-BB, GTN-BT, TT, and HT topologies: Shown are the degrees of freedom (DoF) (left),  compute time (middle) and relative error (right) versus the tolerance $\varepsilon$ used in Algorithm~\ref{alg:gtn-svd}.}
    \label{fig:gtn_svd_comp_1}
\end{figure}
\begin{figure}[t]
    \centering
    \includegraphics[width=\linewidth]{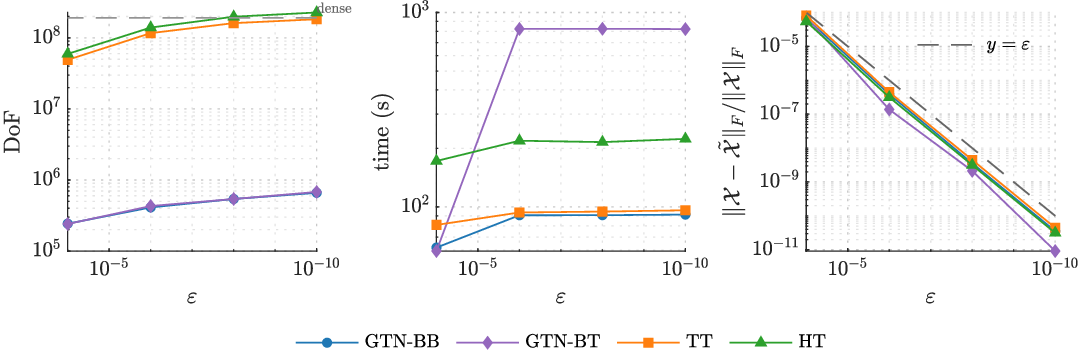}
\caption{Tensor decomposition of the six-dimensional  function \eqref{eq:comp2} using GTN-BB, GTN-BT, TT, and HT topologies: DoF vs. $\varepsilon$ (left); compute time in seconds vs. $\varepsilon$ (middle); relative error vs. $\varepsilon$ (right).}
    \label{fig:gtn_svd_comp_2}
\end{figure}
Figure~\ref{fig:gtn_svd_comp_1} and~\ref{fig:gtn_svd_comp_2} summarize these three metrics for both test functions \eqref{eq:comp1} and \eqref{eq:comp2}, respectively, across the four prescribed tolerances.
For the function $f_1$ in \eqref{eq:comp1}, all four formats remain below the prescribed tolerance, but they do so at significantly different storage cost. At $\varepsilon = 10^{-8}$ GTN-BB requires $47{,}264$ degrees of freedom, whereas TT requires $1.58\times 10^{7}$ and HT requires $1.62\times 10^{7}$; that is, TT exceeds GTN-BB in storage by a factor of $334$, and HT by a factor of $344$. The advantage widens as the tolerance tightens: these factors grow from $232$ and $222$ at $\varepsilon = 10^{-4}$ to $351$ and $358$ at $\varepsilon = 10^{-10}$. Relative to the full tensor, GTN-BB retains $0.025\%$ of the entries at $\varepsilon = 10^{-8}$, a compression factor of approximately $4\times 10^{3}$ compared to the full format, whereas TT and HT still require $8.7\%$ and $8.9\%$ of the dense storage at the tightest tolerance. We also note that, for the GTN results, the degrees of freedom remain fixed at $47{,}264$ for $\varepsilon = 10^{-8}$ and $10^{-10}$; at this rank the representation is already exact to machine precision, so the last two tolerances do not constitute independent data points. In wall-clock time the barbell is also the fastest format, outperforming TT by a factor of $1.4$ and HT by a factor of $3.0$ at $\varepsilon = 10^{-8}$.

For the function $f_2$ in \eqref{eq:comp2}, the tensor formats exhibit similar behavior. At $\varepsilon = 10^{-8}$ GTN-BB requires $5.40\times 10^{5}$ degrees of freedom, whereas TT requires $1.61\times 10^{8}$; that is, TT exceeds GTN-BB in storage by a factor of $299$. This factor grows from $202$ at $\varepsilon = 10^{-4}$ to $299$ at $\varepsilon = 10^{-8}$, then eases to $278$ at $\varepsilon = 10^{-10}$, where TT's storage is approaching the size of the dense tensor itself and therefore cannot grow further. Relative to the full tensor, GTN-BB retains $0.28\%$ of the entries at $\varepsilon = 10^{-8}$, a compression factor of approximately $3.5\times 10^{2}$ compared to the full format, whereas TT requires $84.4\%$ of the dense storage at the same tolerance and $95.8\%$ at the tightest---that is, no compression at all. The HT format fares even worse: although it meets every prescribed tolerance, its storage grows past the size of the dense tensor, requiring $103.9\%$ of the dense storage at $\varepsilon = 10^{-8}$ and $119.1\%$ at $\varepsilon = 10^{-10}$, so that HT exceeds GTN-BB in storage by a factor of $368$ at $\varepsilon = 10^{-8}$. In wall-clock time GTN-BB and TT are tied, differing by a factor of $1.04$ at $\varepsilon = 10^{-8}$, so the barbell is advantageous in storage at no additional cost in time, and it outperforms HT by a factor of $2.4$.

\subsection{Fokker--Planck Equation}

In this section we compute the numerical solution to the Fokker--Planck equation on a four-dimensional torus $\Omega = [0,2\pi]^4$ using GTN-based step rounding algorithms and Fourier pseudo-spectral methods. To this end, let us first consider the stochastic model
\begin{equation}
\label{stochastic-model}
{\text d} x_{i} = \eta\mu_i({\bf x})\,{\text d}t + \sigma\,{\text d}W_i,
\qquad i = 1,\ldots, 4,
\end{equation}
where $(x_1,\ldots,x_4)$ are the phase variables, $\eta$ is
a real number, and $(W_1,\ldots,W_4)$ is a standard vector-valued Wiener process.
In particular we choose
\begin{equation}
\label{vortex-drift}
\begin{aligned}
\mu_1({\bf x}) &= \sin(2x_1)\cos(2x_3), &
\mu_2({\bf x}) &= \sin(2x_2)\cos(2x_4), \\
\mu_3({\bf x}) &= -\cos(2x_1)\sin(2x_3),&
\mu_4({\bf x}) &= -\cos(2x_2)\sin(2x_4),
\end{aligned}
\end{equation}
which defines a divergence--free drift.
As is well known, the corresponding Fokker--Planck equation for 
the probability density function of state vector $\mathbf{\bf x}$ has the form
\begin{equation}
\label{fp-general}
\frac{\partial p({\bf x},t)}{\partial t} =
-\eta\sum_{i=1}^{4}
\frac{\partial}{\partial x_i}\!\left( \mu_i({\bf x})\, p({\bf x},t) \right)
+ \nu \sum_{i=1}^{4}\frac{\partial^2 p({\bf x},t)}{\partial x_i^2},
\qquad \nu = \frac{\sigma^2}{2}.
\end{equation}
Setting $\eta=0$ reduces \eqref{fp-general} to the diffusion equation, while
a non-constant ${\boldsymbol\mu}$ retains the full advection--diffusion structure.

We discretize \eqref{fp-general} in space by Fourier collocation on a uniform grid with $N = 64$ points per axis, so that the dense solution tensor carries $N^d = 64^4 \approx 1.68\times 10^{7}$ degrees of freedom, and we advance in time with a second-order Adams--Bashforth (AB2) integrator, initialized by a single forward-Euler step. GTN integration proceeds by \emph{step truncation}~\cite{rodgers2022}: each step is carried out within the graph format, forming the AB2 increment with the GTN addition and then compressing the result back to a prescribed relative tolerance $\varepsilon = 10^{-8}$ with GTN-Round (Algorithm~\ref{alg:gtn-round}). The two operations consist of addition, which inflates the edges along the traceable path while leaving the remaining edges unchanged, and rounding, which sweeps that same path to compress them back. Together, they return a low rank representation at every step, without ever forming the dense tensor.
We set the initial condition as
\begin{equation}
\begin{aligned}
p(\mathbf{x},0) \;\propto\;
&\ \exp\!\Big(\!-4\big(\cos(x_1-\pi)+\cos(x_3-\pi)-0.5\big)^2
    + 1.5\,\sin(x_1-\pi)\sin(x_3-\pi) \;+\\
&\qquad\quad 1.6\,\cos(x_2-x_4) + 1.2\,\cos\!\big(2(x_2+x_4)\big)\Big),
\end{aligned}
\label{eq:ic}
\end{equation}
normalized to unit mass so that $p({\bf x},0)$ is a valid probability density. 
We represent the solution in the graph tensor network barbell layout (GTN-BB) shown in Figure~\ref{fig:gtn_bb_diff}, and compare the performance of GTN-BB against the TT format. To this end, we record at every step the degrees of freedom of GTN-BB and TT, the step compute time, and the relative $L^2$ error against the dense reference solution.

\begin{figure}[t]
    \centering
    \begin{tikzpicture}[scale=1.15, baseline={(current bounding box.center)}, >=latex]
	\coordinate (N1) at (-2.2, 1); \coordinate (N3) at (-2.2,-1);
	\coordinate (N5) at (-0.8, 0); \coordinate (N6) at ( 0.8, 0);
	\coordinate (N2) at ( 2.2, 1); \coordinate (N4) at ( 2.2,-1);
	\draw[tnbond] (N1)--(N3); \draw[chordbond] (N1)--(N5); \draw[tnbond] (N3)--(N5);
	\draw[tnbond] (N5)--(N6);
	\draw[tnbond] (N2)--(N6); \draw[chordbond] (N4)--(N6); \draw[tnbond] (N2)--(N4);
	\draw[tnleg] (N1) -- ++(135:0.85) node[above left, font=\scriptsize]{$i_1$};
	\draw[tnleg] (N3) -- ++(225:0.85) node[below left, font=\scriptsize]{$i_3$};
	\draw[tnleg] (N2) -- ++(45:0.85)  node[above right,font=\scriptsize]{$i_2$};
	\draw[tnleg] (N4) -- ++(315:0.85) node[below right,font=\scriptsize]{$i_4$};
	\draw[tracepath] (N1) to[bend right=30] (N3);
	\draw[tracepath] (N3) to[bend right=30] (N5);
	\draw[tracepath] (N5) to[bend left=25]  (N6);
	\draw[tracepath] (N6) to[bend left=30]  (N2);
	\draw[tracepath] (N2) to[bend left=30]  (N4);
	\node[physnode, label={[corelbl]above:$\mathcal{G}_1$}] at (N1) {};
	\node[physnode, label={[corelbl]below:$\mathcal{G}_3$}] at (N3) {};
	\node[intnode,  label={[corelbl]above:$\mathcal{G}_5$}] at (N5) {};
	\node[intnode,  label={[corelbl]above:$\mathcal{G}_6$}] at (N6) {};
	\node[physnode, label={[corelbl]above:$\mathcal{G}_2$}] at (N2) {};
	\node[physnode, label={[corelbl]below:$\mathcal{G}_4$}] at (N4) {};
	\matrix (R) [matrix of math nodes, ampersand replacement=\&,
	left delimiter=(, right delimiter=),
	nodes={minimum size=0.72cm, anchor=center, font=\small},
	shift={(7.6,0)}]
	{
		n_1     \& r_{1,3} \& r_{1,5} \& 1      \& 1      \& 1      \\
		r_{1,3} \& n_3     \& r_{3,5} \& 1      \& 1      \& 1      \\
		r_{1,5} \& r_{3,5} \& 1       \& r_{5,6} \& 1      \& 1      \\
		1       \& 1       \& r_{5,6} \& 1       \& r_{2,6} \& r_{4,6} \\
		1       \& 1       \& 1       \& r_{2,6} \& n_2    \& r_{2,4} \\
		1       \& 1       \& 1       \& r_{4,6} \& r_{2,4} \& n_4    \\
	};
	\begin{scope}[on background layer]
		\foreach \i/\j in {1/2,2/1,2/3,3/2,3/4,4/3,4/5,5/4,5/6,6/5}
		\fill[edgegreen!35] (R-\i-\j.north west) rectangle (R-\i-\j.south east);
		\foreach \i/\j in {1/3,3/1,4/6,6/4}
		\fill[chordorange!35] (R-\i-\j.north west) rectangle (R-\i-\j.south east);
	\end{scope}
	\foreach \k/\g in {1/1,2/3,3/5,4/6,5/2,6/4} {
		\node[font=\scriptsize, left=10pt]  at (R-\k-1.west)  {$\mathcal{G}_{\g}$};
		\node[font=\scriptsize, above=6pt] at (R-1-\k.north) {$\mathcal{G}_{\g}$};
	}
\end{tikzpicture}
\caption{The GTN barbell format (GTN-BB) used to represent the solution of the four-dimensional Fokker--Planck equation \eqref{fp-general}. Left: graph tensor network with six cores: four physical ($\mathcal{G}_0,\mathcal{G}_1, \mathcal{G}_2,\mathcal{G}_3$, carrying the physical indices $i_0,i_1,i_2,i_3$, respectively) and two internal ($\mathcal{G}_4,\mathcal{G}_5$). The cores are arranged as two triangles joined by the bridge $(\mathcal{G}_4,\mathcal{G}_5)$. The dashed green arrows trace the traceable path $\mathcal{G}_0 \to \mathcal{G}_2 \to \mathcal{G}_4 \to \mathcal{G}_5 \to \mathcal{G}_1 \to \mathcal{G}_3$ used for the rounding procedure; orange edges correspond to ranks that freeze during addition. Right: the corresponding rank adjacency matrix $\mathbf{R}$, with its rows and columns permuted into traceable order.}
    \label{fig:gtn_bb_diff}
\end{figure}
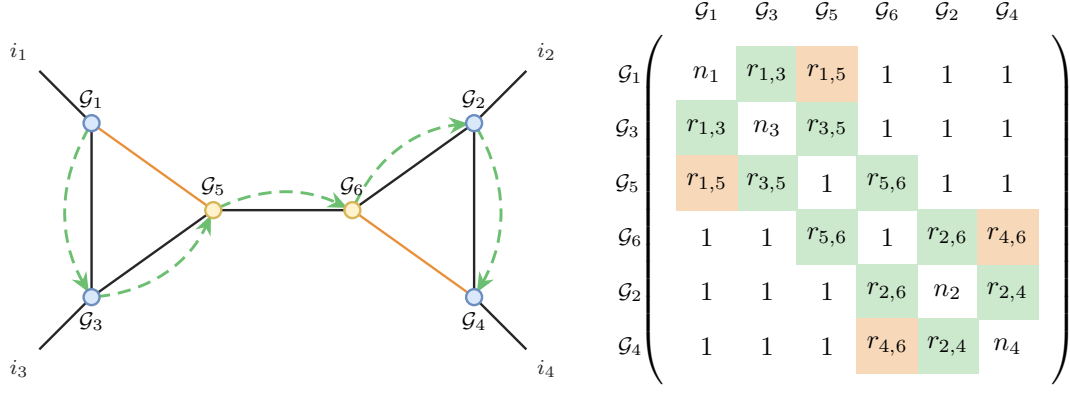

\subsubsection{Zero Drift and Constant Diffusion}
Setting $\eta = 0$ in \eqref{fp-general} reduces the Fokker--Planck equation
 to the classical diffusion equation with periodic boundary conditions,
\begin{equation}
    \dfrac{\partial p(\mathbf{x},t)}{\partial t} = \nu \nabla^2 p(\mathbf{x},t),
\end{equation}

\noindent
which we solve with diffusion constant $\nu = 1$. We use a time step $\Delta t = 2\times 10^{-4}$ for $5000$ steps up to final time $t = 1$.
\begin{figure}[t]
    \centering
    \includegraphics[width=\linewidth]{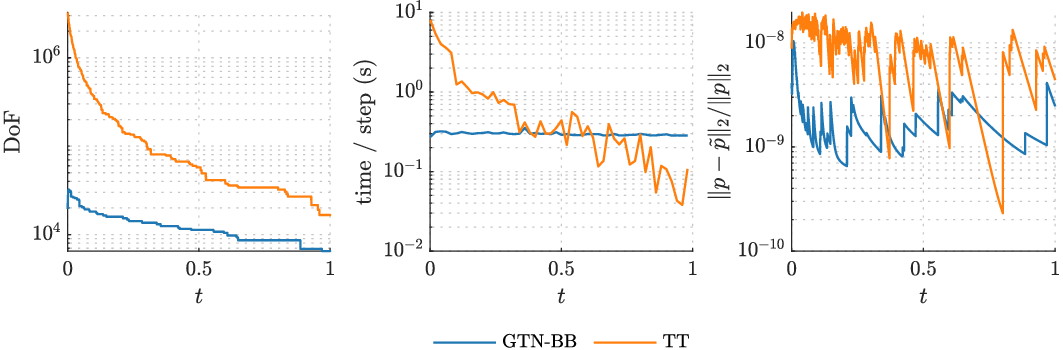}
    \caption{Fokker--Planck equation \eqref{fp-general} with no drift term ($\eta=0$), $\nu=1$, and initial condition \eqref{eq:ic}. Comparison between the performance of GTN-BB and TT representations of the solution in terms of DoF (left), computation time (middle), and relative error (right), all plotted versus time.}
    \label{fig:diffusion_err_comp}
\end{figure}
In Figure~\ref{fig:diffusion_err_comp} we provide a comparison between the performance of GTN-BB and TT representations of the PDF solution in terms of DoF (left), computation time (middle), and relative error (right), all plotted versus time. It is seen that both formats hold the relative error at or below the prescribed tolerance $\varepsilon = 10^{-8}$ for every step: the GTN-BB error never exceeds $1.04\times 10^{-8}$ and ends at $2.46\times 10^{-9}$, while the TT error never exceeds $1.95\times 10^{-8}$ and ends at $4.38\times 10^{-9}$. The two formats therefore deliver similar accuracy, and the comparison rests on the storage and time required to achieve it.

The results display the most notable difference in DoF, where the gap is largest initially. At the first step GTN-BB requires $19{,}528$ degrees of freedom against $3.25\times 10^{6}$ for TT, so that TT exceeds GTN-BB in storage by a factor of $166$. GTN-BB reaches a maximum of $32{,}401$ degrees of freedom at $t \approx 3\times 10^{-3}$, which is $0.19\%$ of the $64^{4} = 16{,}777{,}216$ entries of the full tensor. TT is largest at the initial condition, where it occupies $19.4\%$ of dense. Thereafter both representations shrink as diffusion damps the high-frequency content of the solution, and the advantage narrows monotonically: TT exceeds GTN-BB by a factor of $60$ at $t = 10^{-2}$, $20$ at $t = 10^{-1}$, $5.1$ at $t = 0.5$, and $2.6$ at the final time $t = 1$, where the two formats require $16{,}640$ and $6{,}508$ degrees of freedom respectively. This behavior is expected: as the solution relaxes toward the uniform steady state, its rank collapses in every format, and the particular topology matters less.

The timing results follow a similar pattern initially. Over the full run GTN-BB requires $1460$ s against $2737$ s for TT and $3338$ s for the full reference, so GTN-BB outperforms TT by a factor of $1.87$ and the dense solver by a factor of $2.29$. This aggregate conceals a crossover. Over the first $50$ steps, where the ranks are largest, TT is slower than GTN-BB by a factor of $21.6$, and over the first $250$ steps by a factor of $12.5$. The TT step cost falls from $8.07$ s to $0.06$ s as its ranks decrease, whereas the GTN-BB step cost remains nearly constant at $0.29$ s. Over the final $1000$ steps TT is the faster of the two, by a factor of $2.8$. The advantage of the graph format is thus concentrated in the regime where the solution is high-rank.

\begin{figure}[t]
    \centering
    \includegraphics[width=\linewidth]{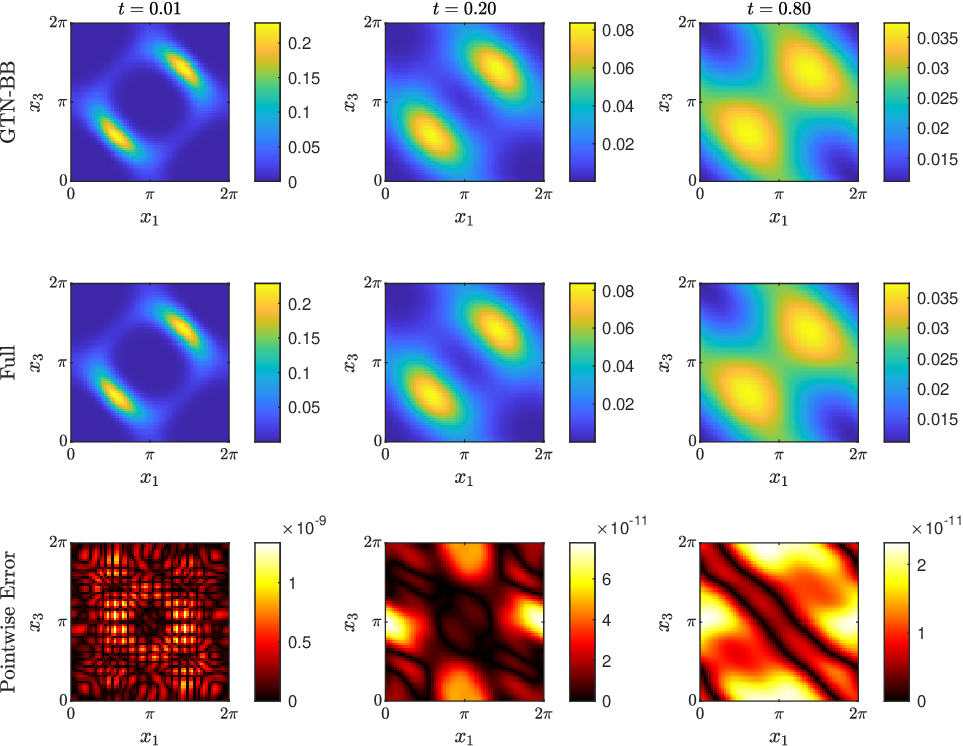}
   \caption{Fokker--Planck equation \eqref{fp-general} with no drift term ($\eta=0$), $\nu=1$, and initial condition \eqref{eq:ic}. We plot the $(x_1,x_3)$ marginals of the solution obtained using GTN-BB with tolerance $\varepsilon=10^{-8}$ at different times, and the error relative to the full (dense) tensor computation.}
    \label{fig:diffusion_comp_1}
\end{figure}

In Figure~\ref{fig:diffusion_comp_1} we plot the $(x_1,x_3)$ marginals of the solution obtained using GTN-BB with tolerance $\varepsilon=10^{-8}$ at different times, and the error relative to the full (dense) tensor computation. It is seen that the GTN-BB solution is indistinguishable from the dense solution at every snapshot, and the pointwise error in the bottom row remains at the level of the prescribed tolerance: over the full integration period, the marginal error never exceeds $6.40\times 10^{-9}$.

\subsubsection{Non-Constant Drift and Constant Diffusion}

Next, we set $\eta = 1$ and $\nu = 0.1$ in \eqref{fp-general}. Since the drift \eqref{vortex-drift} is divergence-free, the advection term in \eqref{fp-general} reduces to $-{\boldsymbol\mu}\cdot\nabla p$. Moreover, since the diffusion coefficient is homogeneous (does not depend on $\bf x$), the time-asymptotic solution of \eqref{fp-general} is a uniform PDF on the torus. We use a time step $\Delta t = 10^{-4}$ for $2\times 10^{4}$ steps up to final time $t = 2$.
\begin{figure}[t]
    \centering
    \includegraphics[width=\linewidth]{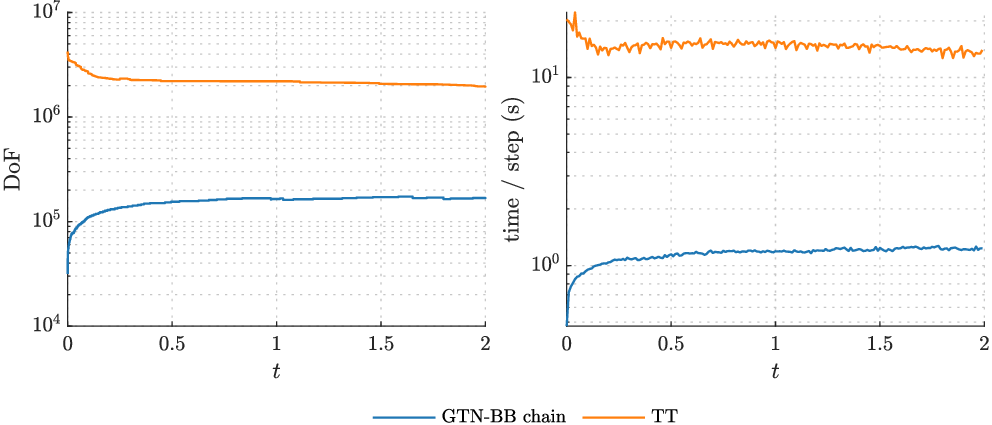}
    \caption{Fokker--Planck equation \eqref{fp-general} with $\eta=1$, $\nu=0.1$, and initial condition \eqref{eq:ic}. Comparison between the performance of GTN-BB and TT representations of the solution in terms of DoF (left) and computation time (right), all plotted versus time.}
    \label{fig:fp_comp}
\end{figure}

\begin{figure}[t!]
    \centering
    \includegraphics[width=\linewidth]{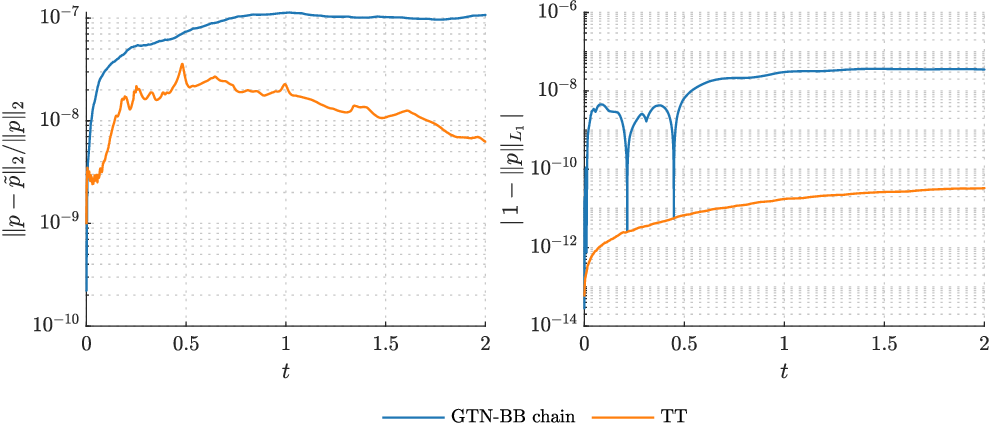}
    \caption{Fokker--Planck equation \eqref{fp-general} with $\eta=1$, $\nu=0.1$, and initial condition \eqref{eq:ic}. Comparison between the performance of GTN-BB and TT representations of the solution in terms of error relative to full (dense) tensor computations (left) and mass conservation (right), all plotted versus time.}
    \label{fig:fp_comp_err}
\end{figure}

\begin{figure}[t!]
    \centering
    \includegraphics[width=\linewidth]{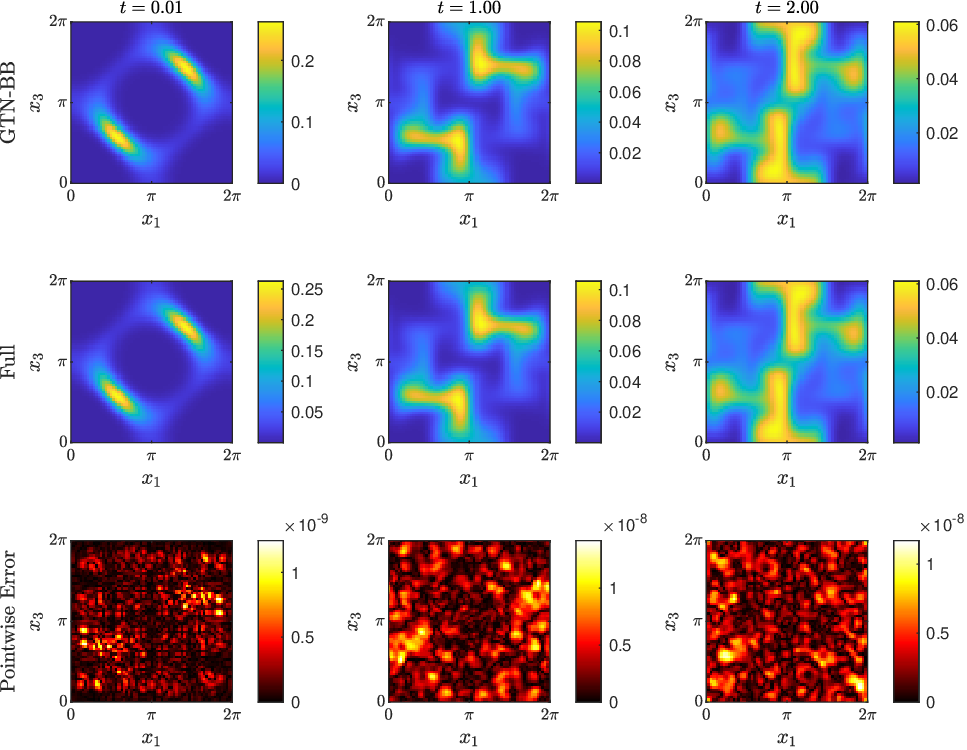}
    \caption{Fokker--Planck equation \eqref{fp-general} with $\eta=1$, $\nu=0.1$, and initial condition \eqref{eq:ic}. We plot the $(x_1,x_3)$ marginals of the solution obtained using GTN-BB with tolerance $\varepsilon=10^{-9}$ at different times, and the error relative to the full (dense) tensor computation.}
    \label{fig:fp_err}
\end{figure}

In Figures~\ref{fig:fp_comp} and \ref{fig:fp_comp_err} we provide a comparison between the performance of GTN-BB and TT representations of the PDF solution in terms of DoF, computation time, relative error, and mass conservation, all plotted versus time.
First, we see that the storage comparison favors the GTN-BB format over TT by a wide margin (more than one order of magnitude). The timing results also display clear separation (approx one order of magnitude) between the performance of the two tensor formats: the mean step costs are approximately $1.16$ s and $13.6$ s respectively.
However, unlike the zero-drift case, the two formats do not deliver comparable accuracy. In particular, at the final time the GTN-BB relative error reaches $1.07\times 10^{-7}$, while the TT error first increases and then decreases to $6.23\times 10^{-9}$. Hence here TT is slightly more accurate than GTN-BB. Both methods exceed the per-step tolerance $\varepsilon = 10^{-9}$, which is expected. Indeed, over $2\times 10^{4}$ steps the truncation error at each step accumulates, and therefore the reported error measures the departure from the dense solution over the whole trajectory rather than the error of a single truncation. The same ordering holds for mass conservation, shown in the right panel of Figure~\ref{fig:fp_comp_err}. The TT solution conserves total mass to $3.31\times 10^{-11}$, whereas GTN-BB conserves it to $3.50\times 10^{-8}$. Both are far below the accuracy of the solution itself, but TT is the more conservative of the two formats.

Finally, in Figure~\ref{fig:fp_err} we plot the $(x_1,x_3)$ PDF marginals of the solution obtained using GTN-BB with tolerance $\varepsilon=10^{-9}$ at different times, and the error relative to the full (dense) tensor computation. The GTN-BB reconstruction is visibly indistinguishable from the dense solution as the initial ring is wound by the vortex. The pointwise error relative to the full (dense) tensor solution remains at the level of the accumulated error. In particular, over the full integration period the pointwise error never exceeds $1.12\times 10^{-7}$.

\section{Summary}
\label{sec:summary}

We introduced a new SVD-based tensor decomposition method for tensor networks with arbitrary graph topologies, extending classical hierarchical SVD-based techniques to networks with cycles and general connectivity. We also introduced addition and rounding procedures for traceable tensor graphs, enabling step-truncation time integration of high-dimensional PDEs directly in graph format, with rank truncation controlled to a prescribed tolerance at every step. We demonstrated the proposed methods on the decomposition of multivariate functions and on the numerical solution of the Fokker--Planck equation, where the graph-format representation attained comparable or better accuracy than the tensor train and hierarchical Tucker formats while using substantially fewer degrees of freedom at lower computational cost.

We emphasize that our results are empirical, and we make no theoretical claim that a graph tensor network (GTN) will outperform the tensor train or hierarchical Tucker in general. Rather, our experiments provide evidence that there are cases in which GTN is more efficient than the fixed path and tree formats, and is at least competitive with them in the others we considered. This leaves considerable room for future work. In its current form the iterative rounding algorithm depends on a heuristic choice of traceable paths through the network. A principled strategy for this selection, together with a sharper truncation criterion on the cyclic edges, would likely improve the rounding algorithm further. Extensions to more complex graphs, higher dimensions, and a wider range of evolution equations are natural directions for future work as well.

\section*{Acknowledgements}
\noindent
This work was supported by the U.S. Air Force Office of Scientific Research (AFOSR), contract number FA9550-24-1-0254.

\bibliographystyle{elsarticle-num}

\bibliography{ref}

@article{grasedyck2010,
  title = {Hierarchical Singular Value Decomposition of Tensors},
  author = {Grasedyck, L.},
  journal = {{SIAM} Journal on Matrix Analysis and Applications},
  volume = {31},
  number = {4},
  pages = {2029--2054},
  year = 2010
}

@article{oseledets2011,
  title = {Tensor-Train Decomposition},
  author = {Oseledets, I. V.},
  journal = {{SIAM} Journal on Scientific Computing},
  volume = {33},
  number = {5},
  pages = {2295--2317},
  year = 2011
}

@misc{zhao2016,
  title = {Tensor Ring Decomposition},
  author = {Zhao, Q. and Zhou, G. and Xie, S. and Zhang, L. and Cichocki, A.},
  year = 2016,
  note = {arXiv:1606.05535}
}

@misc{mickelin2020,
  title = {On Algorithms for and Computing with the Tensor Ring Decomposition},
  author = {Mickelin, O. and Karaman, S.},
  year = 2020,
  note = {arXiv:1807.02513}
}

@article{zheng2021,
  title = {Fully-Connected Tensor Network Decomposition and Its Application to Higher-Order Tensor Completion},
  author = {Zheng, Y.-B. and Huang, T.-Z. and Zhao, X.-L. and Zhao, Q. and Jiang, T.-X.},
  journal = {Proceedings of the {AAAI} Conference on Artificial Intelligence},
  volume = {35},
  number = {12},
  pages = {11071--11078},
  year = 2021
}

@article{wangFCTN2024,
  title = {{SVD}-based Algorithms for Fully-Connected Tensor Network Decomposition},
  author = {Wang, M. and Li, H.},
  journal = {Computational and Applied Mathematics},
  volume = {43},
  number = {5},
  pages = {265},
  year = 2024
}

@article{wangTW2024,
  title = {{SVD}-based Algorithms for Tensor Wheel Decomposition},
  author = {Wang, M. and Cui, H. and Li, H.},
  journal = {Advances in Computational Mathematics},
  volume = {50},
  number = {5},
  pages = {99},
  year = 2024
}

@misc{batselier2018,
  title = {The Trouble with Tensor Ring Decompositions},
  author = {Batselier, K.},
  year = 2018,
  note = {arXiv:1811.03813}
}

@article{htucker,
  title = {Algorithm 941: {Htucker}---A {Matlab} Toolbox for Tensors in Hierarchical Tucker Format},
  author = {Kressner, D. and Tobler, C.},
  journal = {{ACM} Transactions on Mathematical Software},
  volume = {40},
  number = {3},
  pages = {22:1--22:22},
  year = 2014
}

@misc{tttoolbox,
  title = {{TT-Toolbox}: The Git Repository for the {TT-Toolbox}},
  author = {Oseledets, I. V.},
  year = 2014
}

@article{tichavsky2025,
  title = {Optimizing the Order of Modes in Tensor Train Decomposition},
  author = {Tichavsk{\'y}, P. and Straka, O.},
  journal = {{IEEE} Signal Processing Letters},
  volume = {32},
  pages = {1361--1365},
  year = 2025
}

@article{dektor2021,
  title = {Rank-Adaptive Tensor Methods for High-Dimensional Nonlinear {PDEs}},
  author = {Dektor, A. and Rodgers, A. and Venturi, D.},
  journal = {Journal of Scientific Computing},
  volume = {88},
  number = {2},
  pages = {36},
  year = 2021
}

@article{rodgers2022,
  title = {Adaptive Integration of Nonlinear Evolution Equations on Tensor Manifolds},
  author = {Rodgers, A. and Dektor, A. and Venturi, D.},
  journal = {Journal of Scientific Computing},
  volume = {92},
  number = {2},
  pages = {39},
  year = 2022
}

@article{rodgers2023,
  title = {Implicit Integration of Nonlinear Evolution Equations on Tensor Manifolds},
  author = {Rodgers, A. and Venturi, D.},
  journal = {Journal of Scientific Computing},
  volume = {97},
  number = {2},
  pages = {33},
  year = 2023
}

@article{einkemmer2025,
  title = {A Review of Low-Rank Methods for Time-Dependent Kinetic Simulations},
  author = {Einkemmer, L. and Kormann, K. and Kusch, J. and McClarren, R. G. and Qiu, J.-M.},
  journal = {Journal of Computational Physics},
  volume = {538},
  pages = {114191},
  year = 2025
}

@book{risken1996,
  title = {The {Fokker-Planck} Equation: Methods of Solution and Applications},
  author = {Risken, H.},
  publisher = {Springer},
  year = 1996
}

@article{venturi2021,
  title = {Spectral methods for nonlinear functionals and functional differential equations},
  author = {Venturi, D. and Dektor. A},
  journal = {Research in the Mathematical Sciences},
  volume = {8},
  pages = {1-39},
  year = 2021
}

@article{venturi2018,
  title = {The numerical approximation of nonlinear functionals and functional differential equations},
  author = {Venturi, D.},
  journal = {Physics Reports},
  volume = {732},
  pages = {1-102},
  year = 2018
}

@article{rodgers2024,
  title = {Tensor approximation of functional differential equations},
  author = {Rodgers, A.  and Venturi, D.},
  journal = {Physical Review E},
  volume = {110},
  pages = {015310},
  year = 2024
}

@article{boelens2020a,
  title = {Tensor Methods for the {Boltzmann-BGK} Equation},
  author = {Boelens, A. M. P. and Venturi, D. and Tartakovsky, D. M.},
  journal = {Journal of Computational Physics},
  volume = {421},
  pages = {109744},
  year = 2020
}

@book{cercignani,
  title = {The {Boltzmann} Equation and Its Applications},
  author = {Cercignani, C.},
  series = {Applied Mathematical Sciences},
  volume = {67},
  publisher = {Springer},
  address = {New York},
  year = 1988
}

@article{dimarco2014,
  title = {Numerical Methods for Kinetic Equations},
  author = {Dimarco, G. and Pareschi, L.},
  journal = {Acta Numerica},
  volume = {23},
  pages = {369--520},
  year = 2014
}

@article{tang2024,
  title = {Solving High-Dimensional {Fokker-Planck} Equation with Functional Hierarchical Tensor},
  author = {Tang, X. and Ying, L.},
  journal = {Journal of Computational Physics},
  volume = {511},
  pages = {113110},
  year = 2024
}

@inproceedings{li2020,
  title = {Evolutionary Topology Search for Tensor Network Decomposition},
  author = {Li, C. and Sun, Z.},
  booktitle = {Proceedings of the 37th International Conference on Machine Learning},
  volume = {119},
  pages = {5947--5957},
  year = 2020
}

@inproceedings{li2022,
  title = {Permutation Search of Tensor Network Structures via Local Sampling},
  author = {Li, C. and Zeng, J. and Tao, Z. and Zhao, Q.},
  booktitle = {Proceedings of the 39th International Conference on Machine Learning},
  pages = {13106--13124},
  year = 2022
}

@inproceedings{li2023,
  title = {Alternating Local Enumeration ({TnALE}): Solving Tensor Network Structure Search with Fewer Evaluations},
  author = {Li, C. and Zeng, J. and Li, C. and Caiafa, C. F. and Zhao, Q.},
  booktitle = {Proceedings of the 40th International Conference on Machine Learning},
  volume = {202},
  pages = {20384--20411},
  year = 2023
}

@article{zeng2024,
  title = {Bayesian Tensor Network Structure Search and Its Application to Tensor Completion},
  author = {Zeng, J. and Zhou, G. and Qiu, Y. and Li, C. and Zhao, Q.},
  journal = {Neural Networks},
  volume = {175},
  pages = {106290},
  year = 2024
}

\end{document}